\documentclass[sigconf, nonacm]{acmart}
\usepackage{balance}
  


\title{Efficiency and Stability in Euclidean Network Design} 

\author{Wilhelm Friedemann}
\affiliation{
	\institution{Hasso Plattner Institute\\ University of Potsdam}
	\city{Potsdam}
	\country{Germany}
}

\author{Tobias Friedrich}
\affiliation{
	\institution{Hasso Plattner Institute\\ University of Potsdam}
	\city{Potsdam}
	\country{Germany}
}

\author{Hans Gawendowicz}
\affiliation{
	\institution{Hasso Plattner Institute\\ University of Potsdam}
	\city{Potsdam}
	\country{Germany}
}

\author{Pascal Lenzner}
\affiliation{
	\institution{Hasso Plattner Institute\\ University of Potsdam}
	\city{Potsdam}
	\country{Germany}
}

\author{Anna Melnichenko}
\affiliation{
	\institution{Hasso Plattner Institute\\University of Potsdam}
	\city{Potsdam}
	\country{Germany}
}

\author{Jannik Peters}
\affiliation{
	\institution{Hasso Plattner Institute\\University of Potsdam}
	\city{Potsdam}
	\country{Germany}
}

\author{Daniel Stephan}
\affiliation{
	\institution{Hasso Plattner Institute\\University of Potsdam}
	\city{Potsdam}
	\country{Germany}
}

\author{Michael Vaichenker}
\affiliation{
	\institution{Hasso Plattner Institute\\University of Potsdam}
	\city{Potsdam}
	\country{Germany}
}
\usepackage{thmtools}
\usepackage{thm-restate}
\usepackage{hyperref}
\usepackage{cleveref}
\usepackage{subfigure}

\usepackage{tabularx,booktabs,multirow}
\usepackage{mathtools}
\newcolumntype{m}{>{\hsize=.5\hsize}X}
\usepackage{tikz}
\usetikzlibrary{arrows,decorations.pathmorphing,decorations.pathreplacing,backgrounds,positioning,fit,matrix}
\usetikzlibrary{shapes,calc}

\usepackage[ruled]{algorithm}
\usepackage[noend]{algpseudocode}
\usepackage{bbm}
\usepackage{algpseudocode}
\usepackage{algorithm}
\usepackage{pgfplots}
\pgfplotsset{compat=1.10}

\usepackage{multirow}
\usepackage{booktabs} 
\usepackage{float}

\usepackage[ruled,vlined,linesnumbered,algo2e]{algorithm2e}

\usepackage{xspace}

\newcommand{\PG}[2]{\pi_{#1}(#2)}
\newcommand{\dist}{{\mathit d}} 

\newcommand{\SC}{\mathrm{SC}} 
\newcommand{\game}{$\mathbb{R}^d$--GNCG} 
\newcommand{\opt}[1]{\mathit {OPT_{#1}}} 
\newcommand{\apxNE}[1]{#1-NE\xspace} 
\newcommand{\apxNEopt}[2]{(#1, #2)\text{-network}}
\newcommand{\maxDist}{\mathit w_{max}}
\newcommand{\minDist}{\mathit w_{min}}
\newcommand{\Rd}{\mathbb{R}^d}
\DeclareMathOperator{\cost}{cost}
\newcommand{\norm}[1]{\left\lVert#1\right\rVert}

\newtheorem{conjecture}{Conjecture}

\newcommand{\excellone}[1]{
\renewcommand{\arraystretch}{1.5}
  \begin{tabular}[t]{@{}p{4.2cm}@{}}#1\end{tabular}
}
\newcommand{\excelltwo}[1]{
\renewcommand{\arraystretch}{1.5}
  \begin{tabular}[t]{@{}p{5.4cm}@{}}#1\end{tabular}
}
\newcommand{\excellthree}[1]{
\renewcommand{\arraystretch}{1.5}
  \begin{tabular}[t]{@{}p{4.2cm}@{}}#1\end{tabular}
}

\makeatletter
\newcommand\footnoteref[1]{\protected@xdef\@thefnmark{\ref{#1}}\@footnotemark}
\makeatother

\begin{document}

\begin{abstract}
Network Design problems typically ask for a minimum cost sub-network from a given host network. This classical point-of-view assumes a central authority enforcing the optimum solution. But how should networks be designed to cope with selfish agents that own parts of the network? In this setting, minimum cost networks may be very unstable in that agents will deviate from a proposed solution if this decreases their individual cost. Hence, designed networks should be both efficient in terms of total cost and stable in terms of the agents' willingness to accept the network.  

We study this novel type of Network Design problem by investigating the creation of $(\beta,\gamma)$-networks, that are in $\beta$-approximate Nash equilibrium and have a total cost of at most $\gamma$ times the optimal cost, for the recently proposed Euclidean Generalized Network Creation Game by Bilò et al.~\cite{bilo2019geometric}. There, $n$ agents corresponding to points in Euclidean space create costly edges among themselves to optimize their centrality in the created network. Our main result is a simple $\mathcal{O}(n^2)$-time algorithm that computes a $(\beta,\beta)$-network with low $\beta$ for any given set of points. Moreover, on integer grid point sets or random point sets our algorithm achieves a low constant~$\beta$. Besides these results for the Euclidean model, we discuss a generalization of our algorithm to instances with arbitrary, even non-metric, edge lengths. 
Moreover, in contrast to these algorithmic results, we show that no such positive results are possible when focusing on either optimal networks, i.e., $(\beta,1)$-networks, or perfectly stable networks, i.e., $(1,\gamma)$-networks, as in both cases NP-hard problems arise, there exist instances with very unstable optimal networks, and there are instances for perfectly stable networks with high total cost.  
Along the way, we significantly improve several results from Bilò et al. and we asymptotically resolve their conjecture about the Price of Anarchy by providing a tight bound. 	
\end{abstract}

\maketitle

\section{Introduction}
Network Design is a classical and rich research area in Operations Research and Theoretical Computer Science. Core questions in the Network Design literature target how to construct networks with favorable properties like low total cost, high robustness, and good usability. Typically these questions have been addressed as combinatorial optimization problems. Many special cases like the Steiner Tree Problem~\cite{karp1972reducibility}, the Optimum Communication Spanning Tree Problem~(ND7 in~\cite{GJ02}), the creation of Geometric Spanner Networks~\cite{narasimhan2007geometric} or variants of the Network Design Problem~\cite{JLK78} have been thoroughly studied~\cite{MW84,gupta2011approximation,GairingHK14}.

However, almost all previous work on these problems simply assumes that a central authority exists that enforces the optimal network structure obtained by combinatorial optimization. This approach is obviously infeasible for networks with no central governing authority. In such networks the cost of maintaining the network is typically distributed among the network participants, and it can happen that the network with the minimum total cost is not stable, i.e., that selfish participants who own parts of the network prefer to restructure their part. This may be favorable for them individually, but not for the whole network, because this might yield a network with significantly higher total cost. Hence, stability corresponds to reaching an equilibrium in the strategic game that models the interaction of the selfish participants. It can happen that the total cost of any equilibrium of such a game is much higher than the total cost of the optimum network. 

Hence, ideally, in the realm of Network Design, we aim for networks that are efficient in terms of total cost and, at the same time, are as stable as possible. This naturally corresponds to a bi-criteria optimization problem, i.e., finding networks that are $(\beta,\gamma)$-approximate solutions, where $\beta$ is the approximation factor for the cost of the individual agent compared to the induced cost by her best possible strategy and $\gamma$ is the total cost ratio with the minimum possible total cost. Hence in a $(\beta,\gamma)$-network no agent can improve her cost via a strategy change by more than a factor of $\beta$, and the total cost is at most $\gamma$ times the cost of the social optimum state. Extreme cases are $(1,\gamma)$-networks, whose analysis yields bounds for the Price of Anarchy (PoA)~\cite{KP99} or the Price of Stability (PoS)~\cite{ADKTWR,CSSM04}, and $(\beta,1)$-networks that indicate how tolerant the agents have to be to accept the social optimum network~\cite{AL10}. 

In this work we set out to explore the design of $(\beta,\gamma)$-networks for a natural strategic Network Design setting, where a network between nodes that correspond to points in Euclidean space must be established. Each node represents
a selfish agent that strives for centrality in the created network and edges have a cost that is proportional to the Euclidean distance between their endpoints.

\subsection{Related Work}
We will focus our discussion on game-theoretic network formation models. For an overview over classical Network Design from combinatorial optimization, we refer to the surveys by Magnanti and Wong~\cite{MW84} and Gupta and Könemann~\cite{gupta2011approximation}. Also related to our work are results about geometric spanners, in particular, Euclidean spanners, that recently have drawn much attention~\cite{ENS15,CW18,FS20,LS19}. For a good overview over geometric spanners, we refer to the excellent book by Narasimhan and Smid~\cite{narasimhan2007geometric}.

Game-theoretic models for network formation can be divided into variants of the \emph{Network Design Game (NDG)} as introduced by Anshelevich et al.\cite{ADTW08,ADKTWR} and variants of the \emph{Network Creation Game (NCG)} as introduced by Fabrikant et al.~\cite{Fab03}. 

In the NDG a weighted host network is fixed, and agents want to connect subsets of nodes, called terminals, from the network. To do this, the agents decide on payments for the edges of the host network. Thus, the agents buy a sub-network of the host network such that all desired terminal connections are established. If agents can buy arbitrary cost-shares of the edges~\cite{ADTW08},  the PoA is~$n$, where $n$ is the number of agents. Moreover, $(3,1)$-networks exist, i.e., the cost of the social optimum network can be shared among the agents to achieve a $3$-approximate Nash equilibrium. Similar results have been achieved for tailored cost-sharing protocols~\cite{CRV08}, for the node weighted version~\cite{CR09}, and for a version that guarantees connectivity of the formed network~\cite{H09}. For the NDG with fair cost sharing~\cite{ADKTWR} the PoA is $n$, but the PoS is $H_n$, i.e., the $n$-th harmonic number, and the latter bound is tight for the version on directed host networks. Also, Albers and Lenzner~\cite{AL10} showed that instances with only $(\Omega(\log n),1)$-networks exist if all agents want to connect to the same terminal. For the general case instances with only $(\Omega(n),1)$-networks exist. 
Hoefer and Krysta~\cite{HK05} analyzed a geometric version of the NDG, where agents correspond to points in the Euclidean plane. They find that the PoA is $n$ and that the PoS for two agents with two terminals each is $1$. Recently, also a variant with fair cost sharing and topology dependent edge-cost was studied~\cite{BFLMM20}.

Much closer to our model is research on the NCG~\cite{Fab03}, where the agents correspond to nodes of a network and any node can establish undirected links to other agents for the cost of $\alpha$ per link, where $\alpha >0$ is a fixed parameter.
The created network then consists of the union of all links created by the agents. The goal of the agents is to minimize the sum of their cost for creating edges and their average distance to all other agents in the created network, i.e., their closeness centrality~\cite{newman10}. A long line of research~\cite{Fab03,Al06,De07,MS13,MMM15,AM17,BL18,AM18,AM19}
has established that the PoA of the NCG is constant for almost all~$\alpha >0$ and it is conjectured that this holds for all~$\alpha$ \cite{Fab03,ADHL13,MS12,MMM15}. Computing the best possible strategy of an agent in the NCG was shown to be NP-hard~\cite{Fab03} and this also holds for many NCG variants~\cite{MS12,CL15,CLMM16,CLMM17}. However, restricted variants with efficient best response computation also exist~\cite{BG00,ADHL13,Len12,Fried17}. 
Regarding the dynamics, it has been shown~\cite{L11,KL13} that many NCG variants do not have the \emph{finite improvement property (FIP)}~\cite{MS96}. Thus, natural convergence protocols, such as iterated best response dynamics, have no convergence guarantee.

There are only a few works that investigate NCG variants in a geometric setting. Eidenbenz et al.~\cite{EKZ06} consider agents corresponding to points in the Euclidean plane that strategically buy incident edges to create a connected network. Another related geometric game was proposed by Moscibroda et al.~\cite{MSW11}, where the agents pay a fixed price $\alpha >0$ for each edge and aim at minimizing their total stretch, where the stretch is the ratio of the shortest path length in the network and the geometric distance. Guylás et al.~\cite{Gul15} considered a NCG variant where agents correspond to uniformly sampled points in the hyperbolic plane that strive for maximum navigability. Bilò et al.~\cite{BFLLM20} considered a variant with a dynamically changing underlying geometry.
Finally, and closest to our work, Bilò et al.~\cite{bilo2019geometric} recently introduced the \emph{Generalized Network Creation Game (GNCG)} with a given weighted host network. It generalizes the NCG since edges can have weights, and the cost of an edge is defined as $\alpha$ times its weight, for $\alpha > 0$. Besides the version with arbitrary edge weights, variants with arbitrary metric weights, metric weights defined by a tree metric, and metric weights defined by Euclidean distance were proposed. The latter of which is the most natural setting for the creation of communication networks and will be our main focus. The authors of~\cite{bilo2019geometric} prove an upper bound on the PoA of $\left(\frac{\alpha+2}{2}\right)^2$ for the most general version and give a better PoA upper bound of $\frac{\alpha+2}{2}$ for metric weights that is conjectured to hold for the general case as well. For the version with weights defined by a tree metric the bound of $\frac{\alpha+2}{2}$ was shown to be tight. 
For the Euclidean version using the $2$-norm, the setting that is at the heart of our paper, the shown results are far from being tight: a constant lower bound on the PoA that is slightly above $3$ was provided while the upper bound is $\frac{\alpha+2}{2}$. Furthermore, it was shown that the PoA lower bound using the $1$-norm approaches $\frac{\alpha+2}{2}$ if the number of dimensions tends to infinity. Non-trivial bounds on the PoS are given only for the tree metric variant where the PoS is $1$. 
Moreover, it was shown that computing the best possible strategy for an agent is NP-hard for all versions. Regarding the FIP, the status for the Euclidean version with $p$-norm for $p\geq 2$ was left open. For all the other variants, it was shown that the FIP does not hold. Finally, concerning the existence of pure Nash equilibria, it is shown that equilibria exist for weights induced by a tree metric and for suitable $\alpha$ values if all weights are either $1$ or $2$. For the other cases with metric weights, only the existence of $3(\alpha+1)$-approximate equilibria is claimed.    

Besides NCG variants where agents minimize edge costs and distance costs, also models where the cost of an agent depends on her local clustering coefficient~\cite{BK11} or on the edge-connectivity of the created network~\cite{EFLM20} have been proposed. 

\subsection{Model and Notation}
We consider the recently introduced model by Bilò et al.~\cite{bilo2019geometric} for the distributed creation of a network by selfish agents with an underlying geometry, called the \emph{Generalized Network Creation Game (GNCG)}. Our main focus will be on the natural special case of the GNCG where the underlying geometry is Euclidean, called the \emph{Euclidean Generalized Network Creation Game $($\game$)$}.
In this model, a set of $n$ points $P = \{p_1,\dots,p_n\}$ in the $d$-dimensional Euclidean space $\mathbb{R}^d$ is considered, where each point $p_i = (p_{i,1},\dots,p_{i,d})$, for $1\leq i \leq n$, corresponds to a selfish agent. As we will see, all agents will jointly form an undirected weighted network $G = (P,E)$ among themselves, where $P$ is the set of nodes and $E$ is the edge set of $G$. Hence, we will use point, agent and node interchangeably. Moreover, we will use the shorthand $uv$ for denoting the edge $\{u,v\} \in E$. The weight of an edge $uv$ is denoted by $\norm{u,v}$ and it is defined as the metric induced by the $2$-norm\footnote{Our results can be adapted to any $p$-norm. We focus on the $2$-norm for the sake of presentation.}, i.e., for any two points $u=(u_1,\ldots, u_d), v=(v_1, \ldots, v_d)$, we have $\norm{u,v}\coloneqq\norm{u-v}_2= \sqrt{\left(\sum\limits_{k=1}^d(u_k-v_k)^2\right)}$. We will call $\norm{u,v}$ the \emph{length} of the edge $uv$. Let $\maxDist \coloneqq \max_{p_i,p_j \in P}\norm{p_i,p_j}$ and $\minDist \coloneqq \min_{p_i,p_j \in P, p_i \neq p_j}\norm{p_i,p_j}$ denote the longest and shortest distance between any two points in the point set $P$ and let $r \coloneqq \frac{\maxDist}{\minDist}$ be the \emph{aspect ratio} of $P$.
Furthermore for any network $G = (P,E)$ and any $u,v\in P$, let $\PG{G}{u,v} \subseteq E$ be a shortest path from $u$ to $v$ in $G$.

Each agent strategically decides which incident edges to buy in order to minimize her total distance to all other agents in the created network. More precisely, the strategy of agent $u$, denoted as $S_u$, is a subset of nodes in $P\setminus\{u\}$ to which agent $u$ wants to create an undirected edge.
For each $v\in S_u$, we say that $u$ is the \textit{owner} of the undirected edge $uv$. Edges are costly and we assume that the edge cost is proportional to the Euclidean distance between the endpoints of the respective edge. Formally, the price of an edge $uv$ that must be paid by its owner is equal to $\alpha\cdot \norm{u,v}$, where $\alpha>0$ is a fixed parameter of the game.

Any vector $\mathbf{s}=(S_{p_1},\dots,S_{p_n})$, is called a strategy profile. Every strategy profile $\mathbf{s}$ uniquely determines a created network $G(\mathbf{s}) = (P,E(\mathbf{s}))$, where $E(\mathbf{s})=\{p_ip_j\mid p_i\in P, p_j\in S_{p_i}\}$.\footnote{Note that if both $p_i \in S_{p_j}$ and $p_j \in S_{p_i}$ holds, then both agents $p_i$ and $p_j$ would pay $\alpha \norm{u,v}$ for the edge $uv$. As we will see, this cannot happen in any equilibrium.} We will omit the reference to $\mathbf{s}$ when it is clear from the context.

Let $d_G(u,v)$ denote the distance between two nodes $u, v$ in a network $G$, where $d_G(u,v)$ is the sum of the edge lengths of the edges in the shortest $u$-$v$ path in $G$, i.e., $d_G(u,v) \coloneqq \sum_{xy \in \PG{G}{u,v}}\norm{x,y}$. If there is no $u$-$v$ path in $G$, then $d_G(u,v)\coloneqq +\infty$.
We use $d_G(u,U)$ to denote the sum of distances from $u$ to all nodes in $U\subseteq P$ in $G$, and we use $\norm{u,U}$ to denote the sum of the lengths of the edges between $u$ and $U$. 
We call $d_G(u,P)$ the \textit{distance cost} and $\alpha\cdot \norm{u,S_u}$ the \textit{edge cost} of agent~$u$. 
Each agent $u$ aims at minimizing her cost $cost(u,G(\mathbf{s}))$, that is the sum of the agent's edge cost and her distance cost:
$$cost(u,G(\mathbf{s}))\coloneqq \alpha\cdot \norm{u,S_u} + d_{G(\mathbf{s})}(u,P).$$ Note that the parameter $\alpha$ expresses the agents' relative importance of edge costs versus distance costs. 

We measure the \emph{efficiency} of a network $G(\mathbf{s})$ by its \textit{social cost} $\SC(G(\mathbf{s})) \coloneqq \sum_{u \in P} \cost{(u,G(\mathbf{s}))}$. 
The strategy profile $\mathbf{s}^*$ that minimizes $SC(G(\mathbf{s}^*))$ for given points $P$ is called the \textit{social optimum}. We refer to the network $G(\mathbf{s}^*)$ also as the \textit{social optimum network for the point set $P$}, denoted as $\opt{P}$.

An \textit{improving move} for an agent $u$ is a strategy change, that decreases her cost. Agent $u$ plays her \textit{best response} if agent $u$ has no improving move.
A strategy profile $\mathbf{s}$ is a \textit{pure Nash equilibrium (NE)} if all agents play a best response in $G(\mathbf{s})$. 
We say that $\mathbf{s}$ is a \emph{$\beta$-approximate NE (\apxNE{$\beta$})}, if no agent can change her strategy to improve her cost by more than a factor of $\beta$. Moreover, we call strategy profile $\mathbf{s}$ a $(\beta,\gamma)$-NE if $\SC(G(\mathbf{s})) \leq \gamma\cdot \SC(\opt{P})$ and it is a \apxNE{$\beta$}. Strategy profiles induce networks, thus we call $G(\mathbf{s})$ a $(\beta,\gamma)$-network if $\mathbf{s}$ is a $(\beta,\gamma)$-NE. Moreover, if the edge ownership of some network $G$ is not specified, then we assume arbitrary edge ownership when calling $G$ a $(\beta,\gamma)$-network.

We measure the loss of efficiency due to selfishness via the \textit{Price of Anarchy (PoA)}~\cite{KP99} and the \textit{Price of Stability (PoS)}~\cite{ADKTWR,CSSM04}.
Let $worst_P$ (respectively $best_P$) be the highest (respectively the lowest) social cost of any NE on the point set $P$ and let $\mathcal{P}$ be the set of all possible finite point sets in $\mathbb{R}^d$.
Then the PoA is defined as $\sup_{P \in \mathcal{P}}\frac{worst_P}{\SC(\opt{P})}$ and the PoS is $\sup_{P \in \mathcal{P}}\frac{best_P}{\SC(\opt{P})}$.

\subsection{Our Contribution}
We explore the design of networks that at the same time should be efficient in terms of total cost and stable in terms of local changes to the network infrastructure by selfish agents. This very natural focus on approximating both efficiency and stability seems to be novel in the literature on variants of Network Creation Games. Moreover, also in the wider Network Design literature, we are only aware of the works on Network Design Games with cost-sharing on the edges~\cite{ADTW08,CRV08,CR09,H09,AL10} that take a similar point-of-view. 

We study the creation of $(\beta,\gamma)$-networks for the \game~\cite{bilo2019geometric} using the $2$-norm. Such networks are in $\beta$-approximate Nash equilibrium and at the same time have a total cost that is at most $\gamma$ times the optimal total cost. See Table~\ref{table:overviewResults} for a result overview.
\begin{table*}[h]%
	\begin{center}
		\caption{Overview of our results on $(\beta,\gamma)$-networks for the \game \xspace using the $2$-norm.}
		\label{table:overviewResults}
		\renewcommand{\arraystretch}{2.0}
		\begin{tabular}{@{}ccc@{}}\toprule
			\textbf{Socially optimal}	&\centering\textbf{Apx. optimal and apx. stable}					& \textbf{Perfectly stable}  \\
			
			$(\beta, 1)$-networks			&\centering $(\beta,\gamma)$-networks with $\beta, \gamma >1$	&$(1,\gamma)$-networks	\\
			
			\multicolumn{3}{c}{$\xleftrightarrow{\text{\bf efficiency} \hspace*{12.5cm}\text{\bf stability}}$}\\
						
			\excellone{Social optimum computation NP-hard for metric instances [Thm.~\ref{np-alpha-metric-opt}] \\ 
			Instances with only $(\Omega(\sqrt{\alpha}),1)$-networks \mbox{exist} [Thm.~\ref{thm:social_improv}]}
			
			&\excelltwo{Simple $\mathcal{O}(n^2)$ algorithm for computing $\left(\beta,\beta\right)$-networks [Thm.~\ref{thm:apx_graph}]: \\
			 \ $\beta\in\mathcal{O}(1)$ for $\alpha\leq \sqrt[3]{n}$\\
			  \ $\beta\in\mathcal{O}\big(\alpha^{\frac{2x-1}{2x}}+1\big)$ for $\alpha\leq n^x, x\geq 1$ \\
			  \ $\beta\in\mathcal{O}\big(\alpha^{\frac{3x-1}{4x}}+1\big)$ for $\alpha\leq n^x, x\leq 1$\\
			  MST is $(n-1,n-1)$-network [Thm.~\ref{thm:mst-ane}]	\\
			  $(1+\varepsilon,1+\varepsilon)$-network, a.a.s., for random points in $[0,1]^2$, $\alpha\in o(n)$  [Thm.~\ref{thm:apxNE_random}] \\
			  $(2d,2d)$-networks for integer grid point sets in $\Rd$ [Thm.~\ref{theo:d_dim_grids}]}	
			
			&\excellthree{Stability check via computing best responses NP-hard ~\cite{bilo2019geometric}\\
			Existence of stable networks open, no FIP [Thm.~\ref{thm:best-response}]\\
			PoS $>1$ [Thm.~\ref{thm:PoS}]\\			
			PoA $\in \Theta(\alpha)$, $d\to\infty$ [Thm.~\ref{theorem:d_to_infty}]\\
			PoA$\in\Omega\big(\alpha^\frac{2}{3}\big)$, $d\geq 1$ [Thm.~\ref{thm:R1_PoA_LB}]\\
			GNCG $\text{PoA}\in\Theta(\alpha)$ [Cor.~\ref{cor:PoA_GNCG}]	}	\\				
			\bottomrule
		\end{tabular}
	\end{center}
\end{table*}

Our main result is a simple $\mathcal{O}(n^2)$ algorithm that computes a $(\beta,\beta)$-network for $\alpha \leq n^x$ with $\beta\in\mathcal{O}\left(\alpha^{1-\frac{1}{2x}}+1\right)$ for $x \geq 1$ and
$\beta\in\mathcal{O}\left(\alpha^{\frac{3x-1}{4x}}+1\right)$ for $0< x < 1$. See Figure~\ref{fig:apx_plot} for a graphical illustration of these bounds. 
For $\alpha \leq \sqrt[3]{n}$, i.e., if edges are cheap or for large networks, this implies that our algorithm constructs a $(\mathcal{O}(1),\mathcal{O}(1))$-network. We further demonstrate the power of our algorithm by investigating special instance types: grid point sets and uniform random point sets. For them we obtain particularly low constant values for $\beta$ and $\gamma$. Additionally, we also provide $(\beta,\gamma)$-networks with an even simpler construction: we show that any MST on the point set $P$ is a $(n-1,n-1)$-network. 
Using the better outcome of either the MST or the network obtained by our algorithm then yields a $\big(\mathcal{O}(\alpha^{\frac{2}{3}}),\mathcal{O}(\alpha^{\frac{2}{3}})\big)$-network for arbitrary~$\alpha$.
Moreover, we show that the complete network on $P$ is a $(\alpha+1,\frac{\alpha}{2}+1)$-network. This even holds for the GNCG with arbitrary, even non-metric, edge lengths and hence proves that $(\alpha+1)$-approximate NEs always exist. This is an improvement over the claimed bound of $3(\alpha+1)$ in~\cite{bilo2019geometric} and it resolves an open problem from that paper as no bound for the GNCG was provided. 

In contrast to these positive results, we provide negative results for the extreme cases of $(\beta,\gamma)$-networks, i.e., $(\beta,1)$-networks with the optimal total cost and $(1,\gamma)$-networks with perfect stability. For both of these extremes, our results indicate that, unless P=NP, such networks cannot be computed efficiently. In particular, we show that computing a $(\beta,1)$-network is NP-hard for a generalization of the model with metric edge lengths and that $(1,\gamma)$-networks cannot be found via improving response dynamics since the finite improvement property does not hold. The latter was left open in~\cite{bilo2019geometric}. Also, there it was already shown that computing a best possible strategy is NP-hard, which indicates that also deciding stability is a hard problem. 
Moreover, we provide an instance having very unstable networks with optimal total cost, i.e., only $(\Omega(\sqrt{\alpha}),1)$-networks exist, and this instance also shows that the PoS is larger than $1$, i.e., $(1,1)$-networks cannot exist in general. Moreover, we show that $(1,\gamma)$-networks with $\gamma \in \Omega(\alpha^{\frac{2}{3}})$ exist for all $d\geq 1$, i.e., a $\Omega(\alpha^{\frac{2}{3}})$ lower bound for the PoA that significantly improves over the known constant lower bound and which is close to the known $\mathcal{O}(\alpha)$ upper bound. Additionally, we show that for $d$ tending to infinity the PoA is in $\Theta(\alpha)$.
Besides these PoA results for the Euclidean version, we also prove an upper bound of $2(\alpha+1)$ on the PoA for the GNCG with arbitrary, even non-metric, edge lengths. This asymptotically matches the lower bound of $\frac{\alpha+2}{2}$, and it proves a conjecture from~\cite{bilo2019geometric} up to constant factors.

\section{Social Optimum}\label{sec:opt}
We show that minimum cost networks can be rather unstable. 
\begin{restatable}{theorem}{thmSocialImprov}\label{thm:social_improv}
	There exists a set of points where in the unique social optimum network an agent can improve by a factor of at least $\frac{\sqrt{\alpha}}{3}$, i.e., it is a $\big(\frac{\sqrt{\alpha}}{3},1\big)$-network. (See Figure~\ref{img:best-response-cycle}~(left).)
\end{restatable}
\begin{proof}[Proofsketch]
For simplification we allow co-located points. Note that the result still holds asymptotically without co-location, since we can place all co-located points arbitrarily close together.

Consider the three corners of an equilateral triangle with side length $1$, see Figure~\ref{img:best-response-cycle}~(left). 
We place $\frac{n}{3}$ points on each corner. Note that it is without cost to buy all the edges of length $0$ and that two length-1-edges need to be bought or the network will be disconnected. We observe that buying all three length-1-edges gives us a social optimum if  $2 \cdot 2 \cdot \left(\frac{n}{3}\right)^2 > \alpha +  2 \cdot \left(\frac{n}{3}\right)^2$, which is equivalent to $\alpha < 2 \cdot \left(\frac{n}{3}\right)^2$. Given this, we set $n = 3 \lfloor \sqrt{\alpha}+1\rfloor$ and consider the social optimum where every agent buys at most one length-1-edge. Finally, we compute the improvement factor for one agent by selling a length-1-edge resulting in
\[
	\frac{\alpha+2\frac{n}{3}}{3\frac{n}{3}} \geq \frac{\alpha+2\sqrt{\alpha}}{3\sqrt{\alpha}+3} \geq \frac{\sqrt{\alpha}}{3}.\qedhere
\]
\end{proof}
\noindent Next, we show that computing the social optimum network is NP-hard for any fixed $\alpha$ in the more general metric version of the GNCG (M-GNCG)~\cite{bilo2019geometric}.
There, agents are nodes of a given complete weighted host network $H=(V,E(H))$ with edge weights $w\colon V\times V\rightarrow \mathbb{R}_+$ satisfying the triangle inequality. Hence, the edge price of $uv\in E(H)$ is $\alpha\cdot w(u,v)$.

\begin{restatable}{theorem}{thmOPT}\label{np-alpha-metric-opt}
	For any $\alpha > 0$, computing a social optimum  in the M-GNCG is NP-hard.
\end{restatable}
\begin{proof}
	\begin{figure}[h]
		\centering
\includegraphics[scale=0.65]{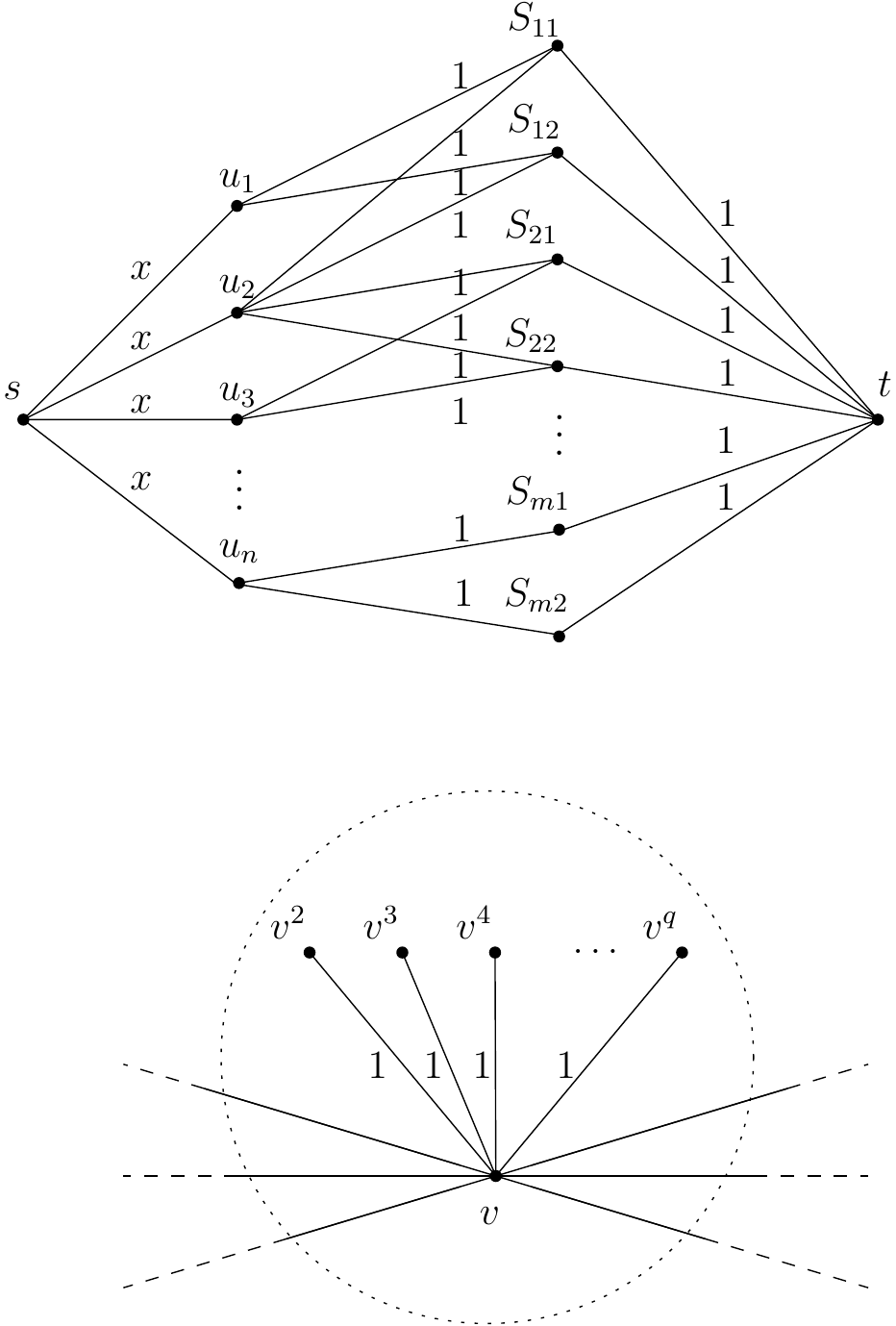}		
		\caption{Illustration of the reduction for $c=2$. Each node $v\in V$ in the top network forms a star with the corresponding nodes $v_2, \dots v_q$ as depicted on the bottom figure. 
		}
		\label{img:np_social_opt}
	\end{figure}

	We perform the reduction from the \textsc{HittingSet(HS)} problem:
	  Given a set of elements $U = \{u_1, \cdots , u_n\}$ and a collection of sets $\mathcal{S}=\{S_1,\dots,S_m\} \subseteq \mathcal{P}(U)$, the problem is to find a minimum hitting set $\mathcal{H} \subseteq U$, such that each set is hit, i.e.,  $\forall S \in \mathcal{S}: S \cap \mathcal{H} \neq \emptyset$.
	
	Consider the corresponding instance of the social optimum problem. 
	We define a host network $H=(V,E)$ such that each element $u \in U$ corresponds to a one node in $V$,  and there is one node for every set $S \in \mathcal{S}$. We connect set nodes with element nodes if their corresponding set contains the corresponding element. We create two other nodes $s$ and $t$ in $V$ that are adjacent to element nodes (resp. set nodes).  Hence, element nodes adjacent to $s$ indicates that the corresponding elements are in the hitting set. To handle all $\alpha$, we duplicate set nodes $c$ times and connect each node to $q-1$ additional leafs $v^2,\cdots,v^q$, i.e., inflate all nodes to stars, for some integer $c$ and $q$. 
	See Figure~\ref{img:np_social_opt} for an illustration.

	More formally, 
	let $H=(V,E)$ be a complete host network such that $V=V_1\cup V_2$, $E=E_1\cup E_2$, where \[V_1= \{s,t\}\cup\bigcup\limits_{i=1}^n \{u_i\}\cup\bigcup\limits_{i=1}^{m}\bigcup\limits_{j=1}^c \{s_{ij}\}, \ \ V_2 = \bigcup\limits_{v\in V_1}\bigcup\limits_{i=2}^{q} \{v^i\}\]
	\[
	E_1=\bigcup\limits_{i=1}^n \{su_i\} \cup \bigcup\limits_{i=1}^n \bigcup\limits_{j=1}^c\{u_is_{pj}\mid u_i\in S_p\}\cup \bigcup\limits_{i=1}^m\bigcup\limits_{j=1}^c \{s_{ij}t\} \cup \bigcup\limits_{v\in V_1}\bigcup\limits_{i=2}^{q} \{vv^i\}\]
	\[
	 E_2 = \{V\times V\} \setminus E_1,
	\]
	where $p$ is an index between $1$ and $m$, $q\in\mathbb{N}$ the number of nodes in each star, and $c\in\mathbb{N}$, the number of set nodes duplications, will be specified later.
	We assume that the weight of each edge between $s$ and an element node $u_i$ is $w(s,u_i)=x$, while all other edges in $E_1$ are of length 1. All edges in $E_2$ are the metric closure for the subnetwork $(V,E_1)$, i.e., for any $xy\in E_2$, $w(x,y)=d_{(V,E_1)}(x,y)$. 
	We choose $q=1+\big\lceil \frac{\sqrt{\alpha}}{2}\big\rceil$, $x=2+\frac{4q^2}{\alpha}$ and $c = 1+\big\lceil\frac{\alpha x}{4q^2}\big\rceil$. In the following, we show that the edges in the optimum network $\opt{V}$ incident to $s$ induce a minimal \textsc{HS}.

	We start by proving that all edges of length $1$ are in the optimum. 
Note that $\opt{V}$ does not contain any edges from $E_2$ since it is always beneficial to have edges from a shortest path rather than one edge that is its metric closure. 
Hence, all star edges, i.e., edges connecting $V_1$ and $V_2$, are in  $\opt{V}$.
For the other edges, we observe that $(V,E_1)$ is bipartite. Hence, if some length-1-edge $xy\in V_1\times V_1$ is not in $\opt{V}$, then adding $xy$ improves the distance between $2(q-1)$ leafs adjacent to the star centers $x$ and $y$ by at least 2, i.e., the total distance in $\opt{V}$ would increase by at least $4q^2$. Since $4q^2 > 4\frac{\alpha}{4} = \alpha$, the edge $xy$ is in $\opt{V}$.

	Next, we prove by contradiction that every set node will be hit. 
	This means, that for every set $S_i$, network $\opt{V}$ contains at least one of the length-$x$-edges $(s,u_j)$ such that $u_j$ is an element of $S_i$. Let $S_i$ be a set, which is not hit and $u_j\in S_i$ is one of its elements. Note that at least one set is hit, otherwise the network would not be connected.
	Then adding $su_j$ costs $\alpha \cdot x$ but shortens the distance between $2cq^2$ nodes by 2:  between the $(q-1)$ leaf nodes adjacent to $s$ and the $c\cdot (q-1)$ leafs adjacent to each node $s_{i1},\ldots,s_{ic}$, as well as between the star centers. Since $\alpha x =4q^2\frac{\alpha x}{4q^2} < 4cq^2$, it is beneficial to add the edge.
	
	Finally, we need to show that the number of edges between $s$ and the element nodes is minimal, i.e., that $\opt{V}$ corresponds to the minimum  hitting set. 	Denote the number of length-$x$-edges in $\opt{V}$ as $k$. 	We calculate the social cost of the optimum.
	
	Let $\Delta$ be the sum of costs of all length-1-edges and the distances between all nodes except the distances between $s$, elements from $U$, and their corresponding leaf nodes.
	Note that all sets in $\mathcal{S}$ are hit by the construction, and that a shortest path between two nodes $x, y\in V$ does not include node $s$ unless $s$ is one of the two nodes $x$ or $y$, since $x > 2$.
	Thus, $\Delta$ depends only on the instance and not on $k$.
	The distance between $s$ and the element nodes is either $x+2$ or $x$ if nodes are directly connected. Thus, the social cost of the network is 
	\begin{align*}
		&= k\alpha x+2kq^2x+2(n-k)q^2(x+2)+\Delta  \\
		&=\alpha kx+ 2nq^2(x+2)-4kq^2+\Delta                                 \\
		 &=\alpha k\left(2+\frac{4q^2}{\alpha}\right)-4kq^2+2nq^2(x+2)+y  =2k\alpha + 2nq^2(x+2)+\Delta
	\end{align*}
	Clearly, the social cost is minimal when $k$, the size of the hitting set, is minimal. 
\end{proof}

\noindent 
Clearly, hardness for a problem on metric instances does not imply hardness for Euclidean instances. However, given that many variants of minimum weight Euclidean t-spanner problems are also NP-hard (e.g., see \cite{carmi2013minimum} and the references therein), and since these problems seem to be very close to computing a social optimum network, we think it could be possible to either adapt our reduction for the metric case via suitable gadgets or to reduce from hard minimum weight t-spanner problems directly. Therefore, we conjecture the following:
\begin{conjecture}
 Computing a social optimum network in the \game \xspace is NP-hard.
\end{conjecture}

\section{Nash Equilibrium}\label{sec:NE}
We investigate the existence and the computation of (approximate) NEs. We show that iteratively playing best responses is not guaranteed to lead to a NE and give sufficient conditions for their existence. Then follows the main result of our paper: a simple and efficient algorithm for computing a $\apxNEopt{\mathcal{O}(\alpha^{2/3})}{\mathcal{O}(\alpha^{2/3})}.$
\subsection{Existence}\label{sec:NE_existence}
An obvious way of finding a NE would be to iteratively play best responses. However, the following theorem shows that doing so does not necessarily lead to a NE.
\begin{restatable}{theorem}{thm:best-response}
	The \game \xspace with $d \geq 2$ does not have the finite improvement property.
	\label{thm:best-response}
\end{restatable}
\begin{proof}[Proofsketch]
	We prove this statement by providing a best response cycle, i.e., a cyclic sequence of networks obtained by iterative strategy changes to best responses, in $\mathbb{R}^2$ for $\alpha=1$. See  Figure \ref{img:best-response-cycle}~(right) for illustrations of the steps of the cycle.
	\begin{figure}[t]
		\centering
		\includegraphics[width=\linewidth]{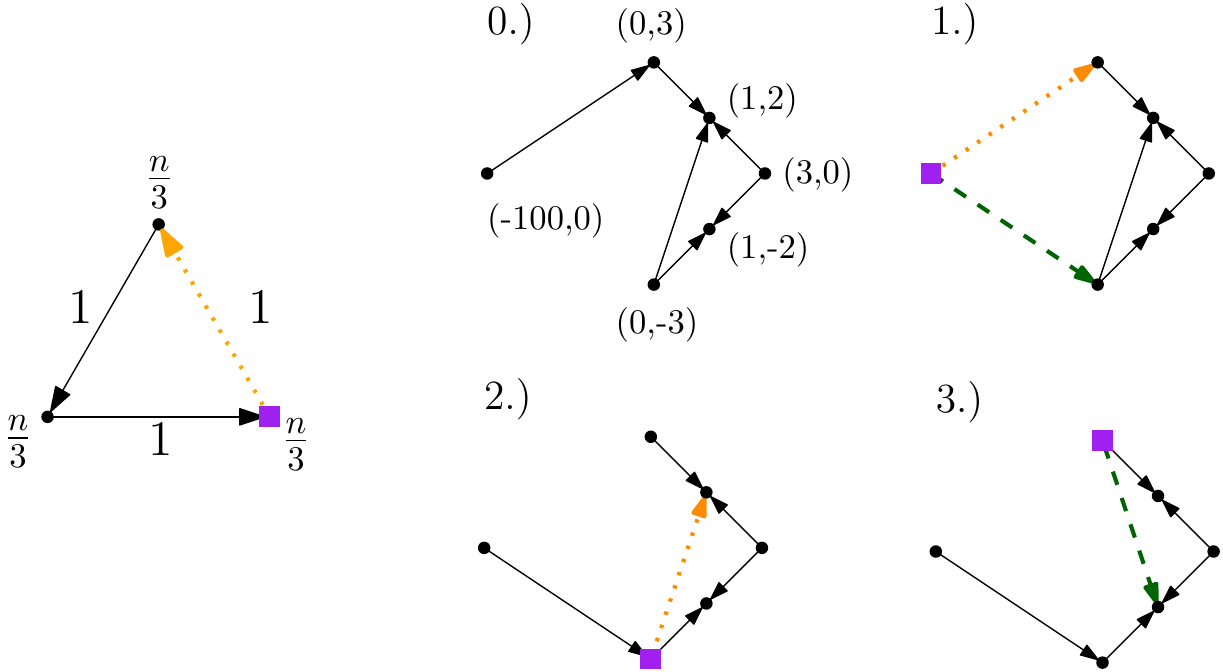}
		\caption{Left: Sketch of the optimum network from Theorem~\ref{thm:social_improv}. In each corner there is a node-cluster of size $\frac{n}{3}$ arbitrarily close together. The improving move removes the dotted orange length-1-edge.
		Right: Best response cycle for the \game \xspace for $d=2$ and $\alpha = 1$. In each step the marked agent may buy edges (dashed green line) and/or sell edges (dotted orange line). The network after step $3$ is the same as the initial network, but mirrored. Edge directions indicate ownership. Edges point away from their owners.}
		\label{img:best-response-cycle}
	\end{figure}
\end{proof}
\noindent As a first step towards showing that in some cases Nash equilibria exist, we show that, if $\alpha$ is large enough, any center sponsored star, where the center buys all edges, is a NE. 
\begin{restatable}{lemma}{lemStartNE}\label{lem:start_NE}
	Let $T$ be the star of a point set $P$ centered at the node $n_0$, such that $n_0$ owns all edges. 
	If $\alpha \ge \max_{u,v \in P, u \neq v} \frac{\norm{u, n_0} + \norm{n_0, v}}{\norm{u,v}} - 1$, then $T$ is a NE.
	
\end{restatable}
\begin{proof} 
	Since $n_0$ owns all edges, she does not want to buy, sell, or swap any additional edges. Therefore we only need to check, whether any of the non-center nodes wants to buy an edge.
	
	We consider the case, where node $v$ buys an edge towards node $u$. Due to the triangle inequality, this would only improve the distance towards $u$ and it would not change any of the other distances. However, since, by assumption, 
	$(\alpha + 1)\norm{u,v} \ge \norm{u, n_0} + \norm{n_0, v}$
	 buying this edge does not decrease the cost. 
	Therefore $T$ is a NE.
\end{proof}

\begin{corollary}\label{lemma:star-NE}
	If $\alpha \ge 2r-1$, for aspect ratio $r$, then every center sponsored star is a NE.
\end{corollary}
\noindent Using this corollary we can now show that a uniform random point set asymptotically almost surely has a Nash equilibrium if $\alpha$ is asymptotically larger than $n$. 
\begin{restatable}{theorem}{thmRandSqNE}
	Let $P_n =\lbrace v_1, \dots, v_n\rbrace\subseteq [0,1]\times [0,1]$ be  $n$ points chosen uniformly at random and let $(\alpha_n)_{n \in \mathbb{N}}$ be a sequence of positive real numbers. If $\alpha_n \in \omega(n)$, then $P_n$ asymptotically almost surely (a.a.s.) has a NE for any $\alpha \ge \alpha_n$.
\end{restatable}
\begin{proof}
	By Corollary~\ref{lemma:star-NE} it suffices to show that the aspect ratio of $P_n$ is upper bounded by $\frac{\alpha + 1}{2}$. Since the maximum possible distance is $\sqrt{2}$, it is enough to show that the closest pair of points are at least $\frac{2\sqrt 2}{\alpha + 1}\eqqcolon d_\alpha$ apart.
	For any $i = 1, \dots, n$ we define the random variables 
	\[
	X_i = \begin{cases}
		1, \text{if } \min_{v_j \in P_n\setminus\lbrace v_i \rbrace} \norm{v_i,v_j} \le d_\alpha \\
		0, \text{ otherwise. }
	\end{cases}
	\]
	Also, let $X = \sum_{i=1}^n X_i.$ We thus want to show that a.a.s. $X = 0$. 
	
	First, we observe that the probability of a single point falling into the $d_\alpha$ neighbourhood of $v_i$ is at most $\pi d_\alpha^2$. 
	Thus, $\mathbb P[X_i = 0] \ge (1- \pi d_\alpha^2)^n$ and \sloppy $\mathbb E[X_i] = \mathbb P[X_i = 1] \le 1-(1- \pi d_\alpha^2)^n$. 
	
	Now we can apply Markov's and Bernoulli's inequalities to bound the probability: 
	\begin{align*}
	\mathbb P[X \neq 0] &= \mathbb P[X \ge 1] \le \mathbb E[X] \le n(1-(1- \pi d_\alpha^2)^n)\\
	&\le n^2\pi d_\alpha^2	= \frac{8 \pi n^2}{(\alpha + 1)^2}.
	\end{align*}
	Since $\alpha \ge \alpha_n$ and $\alpha_n \in \omega(n)$, this probability approaches zero as $n$ tends to infinity. Therefore, $P_n$ asymptotically almost surely has a Nash equilibrium. 
\end{proof}

\subsection{Approximation}\label{sec:NE_apx}
We show how to construct $(\beta,\gamma)$-networks for different values of $\alpha$.
First, we prove a general upper bound on the approximation factor.
\begin{restatable}{theorem}{thmApxNEClique}\label{thm:apxNE_clique}
	Let $P$ be a set of points in $\Rd$. Any complete network $K = (P, P\times P)$ is a $\apxNEopt{\alpha+1}{\frac{\alpha}{2}+1}$ in the \game.
\end{restatable}
\begin{proof}
	Since $K$ is a complete network, every agent $u$ can improve its strategy only by deleting its edges.
	Let $K'$ be a network obtained after a strategy change.
	Since the deletion of edges increases the distance cost,  $\cost(u,K')\geq \dist_{K'}(u,P)\geq \dist_K(u,P)$.
	In the worst case, $u$ owns all its incident edges in $K$.
	Hence, $\cost(u,K)\leq \alpha \cdot  d_K(u,P) + d_K(u,P)$, and we get
	\[
		\frac{\cost(u,K)}{\cost(u,K')} \leq \frac{\alpha \cdot d_K(u,P)+d_K(u,P)}{d_K(u,P)} = \alpha + 1.
	\]

	In a similar way we can prove that the social cost of $K$ is at most $\left(\frac{1}{2}\alpha+1\right)$ times the social cost of $\opt{P}$.
	By the triangle inequality, for any edge $uv$ in $K$, $\norm{u,v}\leq d_{\opt{P}}(u,v)$.
	Hence, $\SC(K) =  \frac{1}{2}\alpha\sum\limits_{u,v\in V}\norm{u,v}+\sum\limits_{u,v\in V}\norm{u,v}\leq \left(\frac{1}{2}\alpha+1\right)\sum\limits_{u\in P}d_{\opt{P}}(u,P)$, while the social cost of the social optimum is at least its distance cost $\sum\limits_{u\in P}d_{\opt{P}}(u,P)$.
	Therefore, $\SC(K)\leq \left(\frac{1}{2}\alpha+1\right)\SC(\opt{P})$.\qedhere
\end{proof}
\noindent In the following we show that it is possible to construct a $(\beta,\gamma)$-network with $\beta,\gamma \in o(\alpha)$.
We call $S$ a \emph{$k$-degree $t$-spanner on $P$} if all its nodes have degree at most $k$ and for any two points $p_1, p_2\in P$, $d_S(p_1,p_2)\leq t\cdot\norm{p_1,p_2}$.\footnote{\label{fn:k_distr_spanner} Algorithm~\ref{alg:apxNE} and all the following results hold for a more general setting. We call $S$ a \emph{$k$-distributable} $t$-spanner if it is possible to assign all edges of $S$ to agents such that each agent owns at most $k$ edges. Then we can generalize Algorithm~\ref{alg:apxNE} by constructing a $k$-distributable $t$-spanner in step 4 and 9.} Consider Algorithm~\ref{alg:apxNE}.
\begin{algorithm2e}[ht]
	\SetKwInOut{Input}{input}
	\Input{ $n$ points $P$ in $\mathbb{R}^d$, parameters $k\in \mathbb{N}, t\in\mathbb{R}_{> 1}, b\in\mathbb{R}_{\ge 1}$, $0\le c\le n-1$}
	\For{$v\in P$}{
		$B_v\gets \lbrace u \in P \mid \norm{u,v} \le \frac{\maxDist}{b} \rbrace $ and $C_v\gets \lbrace u \in P \mid \norm{u,v} \le 2\frac{\maxDist}{b}  \rbrace $\; 
	}
	\uIf {there is a node $v \in P$ with $\lvert P\setminus B_v\rvert < c $}{
		construct a $k$-degree $t$-spanner $G$ on $C_v$\;
		assign the edge ownership such that each agent owns at most $k$ edges in $G$\;
		\For {$u \in P\setminus C_v$} {
			Find node $u'\in C_v$ closest to $u$ and add $u'$ to the strategy of agent $u$\;
		}
	}
	\Else{
		construct a $k$-degree $t$-spanner $G$ on $P$\;
		assign the edge ownership such that each agent owns at most $k$ edges in $G$\;
	}
	\caption{$\mathcal{O}(n^2)$ algorithm for computing a $(\beta,\beta)$-network in the \game. }\label{alg:apxNE}
\end{algorithm2e}
The idea of this algorithm is simple: if $P$ has a large cluster of closely located points, a spanner for the cluster points (set $B_v$) and all close points (set $C_v$) is built. The rest of the points is conneced with the shortest edges to $C_v$. (See Figure~\ref{fig:apxNE_alg}~(left).) 
If the set of points $P$ is sparsely distributed, i.e., there is no cluster, a spanner for the entire set of points is constructed. (See Figure~\ref{fig:apxNE_alg}~(right).) 
\begin{figure}[h]
	\centering
	\includegraphics[width=\linewidth]{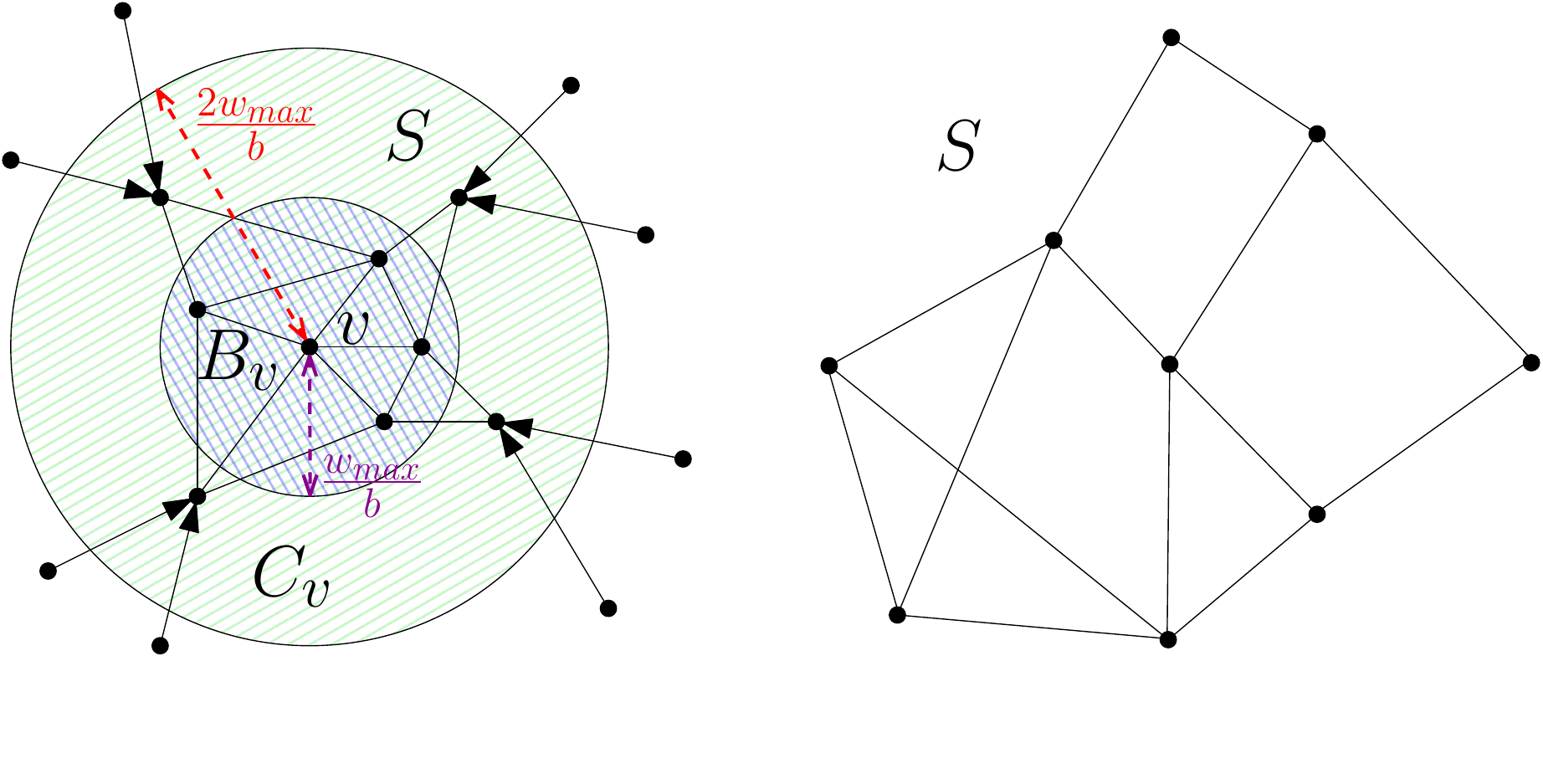}	
	\caption{Illustration of the network computed by Algorithm~\ref{alg:apxNE}. Left: The network if there exists a node $v$ with at most $c$ points with distance at least $\frac{\maxDist}{b}$.  Right: a bounded degree spanner produced by the algorithm in case all points are sparsely distributed.
	Edges point away from their owners. Undirected edges can be assigned arbitrarily as long as every agent owns at most $k$ edges.}
	\label{fig:apxNE_alg}
\end{figure}

\noindent We now prove our main results, i.e., the approximation bounds achieved by constructing the network via Algorithm~\ref{alg:apxNE}.
\begin{restatable}{theorem}{thm:apx_NE}\label{thm:apx_NE}
	Let $b\in \mathbb{R}_{\geq 1},  0\le c\leq n-1$, and let $k\in \mathbb{N}, t\in \mathbb{R}_{>1}$ be parameters such that we can construct a $k$-degree $t$-spanner for $P$.
	Algorithm~\ref{alg:apxNE} computes a $\beta$-NE with $$\beta =\max\left\{ \frac{k b}{c}\alpha +t, \frac{4 k}{b}\alpha  + 2t+1, \frac{2\alpha}{n-c}+2, \frac{4c(b + 2t)}{n-c}+6t\right\}.$$
\end{restatable}
\begin{proof}
 
We divide the proof into two parts corresponding to each if-case of the algorithm. 
First, we prove that if all nodes have at least $c$ nodes at distance of at least $\frac{\maxDist}{b}$, i.e., for all $v\in P$, $\lvert P\setminus B_v\rvert\geq c$, then the $k$-degree $t$-spanner $G$ is a $\left(\frac{kb\alpha}{c}+t\right)$-NE.

Consider an agent $u\in P$ playing a strategy $S_u$ in $G$. 
We need to evaluate the maximal improvement of the cost function that can be made by $u$ after changing its strategy from $S_u$ to $S'_u$. 
Let $G'=(P,E')$ be the network obtained after the improving move.  
Since $G$ is a $t$-spanner and $u$ owns at most $k$ edges of weight $\maxDist$, the cost of agent $u$ before the move is 
\begin{align*}
\cost(u, G) &= \alpha\cdot \norm{u,S_u}+ \dist_G(u,P)\leq \alpha k\cdot\maxDist + t\cdot \norm{u,P}\\
&\leq \alpha k\cdot\maxDist + t\cdot \dist_{G'}(u,P).
\end{align*}
By construction, $u$ has at least $c$ nodes at distance at least $\frac{\maxDist}{b}$. 
Hence, $\frac{b}{c\cdot \maxDist}\dist_{G'}(u,P)\geq  1$.
Combining with the inequality above, we get $\cost(u,G)\leq \frac{\alpha k b}{c}\cdot\dist_{G'}(u,P) + t\cdot\dist_{G'}(u,P)$. 
With this we can evaluate the maximal improvement made by $u$:
\begin{align}
\frac{\cost(u,G)}{\cost(u,G')}
&\leq \frac{\frac{\alpha k b}{c}\cdot\dist_{G'}(u,P) + t\cdot\dist_{G'}(u,P)}{\alpha\cdot \norm{u,S'_u}+\dist_{G'}(u,P)}\nonumber\\
&\leq \frac{\left(\frac{\alpha k b}{c} + t\right)\cdot\dist_{G'}(u,P)}{\dist_{G'}(u,P)}
= \frac{\alpha k b}{c }+t.
\label{ineq:apxNE_case1}
\end{align}

Now we analyze the case when there is a point $v\in P$ such that $|P\setminus B_v| < c$. 
In this situation, the algorithm computes a network $G=(P,E)$ that contains a spanner on $C_v$ with attached leaf nodes from $P\setminus C_v$ 
(see Figure~\ref{fig:apxNE_alg}~(left) for an illustration). 
In the following part of the proof we show that the improvement factor of every agent $u\in P$ is bounded as well.
We distinguish two cases depending on whether $u\in C_v$ or $u\in P\setminus C_v$.

If $u\in C_v$, we observe that, by construction, $u$ owns at most $k$ edges in $G$. 
Let $S_u$ be the strategy of $u$ in $G$, let $S'_u$ be the new improving strategy, and let $G'$ be the network obtained after the improving move. 
Since $u$ only buys edges to the nodes from $C_v$ in $G$, each of her edges has length of at most $\frac{4\maxDist}{b}$. 
 Then the total edge cost of $u$ is at most $4\alpha k\frac{\maxDist}{b}$. 
 Note that $\dist_{G'}(u,P)\geq \sum_{w\in P}\norm{u,w}\geq\maxDist$. Thus, the edge cost of $u$ is at most $\frac{4\alpha k}{b}\cdot \dist_{G'}(u,P)$. 
 
 The distance cost for $u$ is $d_G(u,P)= \dist_G(u,C_v) + \dist_G(u, P\setminus C_v)$. 
 Clearly, $\dist_G(u,C_v)\leq t\cdot \norm{u,C_v}\leq t\cdot \dist_{G'}(u,C_v)$. 
 To analyze the second term $\dist_G(u, P\setminus C_v)$, consider a node $x\in P\setminus C_v$.
 Let $y\in C_v$ be a node to which $x$ buys an edge. The existence of $y$ follows from the construction of $G$. 
 By the triangle inequality, $\dist_G(u,x)\le \dist_G(u,y)+\norm{y,x}$. 
 Since $G$ is a $t$-spanner for $C_v$, this implies that $\dist_G(u,y) \le t\cdot \norm{u,y} \leq t(\norm{u,x}+\norm{y,x})\leq 2t\cdot \norm{u,x}$, 
  where the last inequality holds since, by construction, $\norm{x,y}\leq \norm{x,u'}$ for any $u'\in C_v$.  
 Therefore, $\dist_G(u,x) \le 2t\cdot \norm{u,x}+\norm{y,x}\le (2t+1)\cdot \norm{u,x}$. 
 Finally, we can evaluate the maximum improvement that can be made by agent $u$: 
\begin{align}
\frac{\cost(u,G)}{\cost(u,G')} &\leq \frac{\alpha \norm{u,S_u} + \dist_G(u,P)}{\alpha \norm{u,S'_u} + \dist_{G'}(u,P)}\nonumber\\
&\le \frac{\frac{4\alpha k}{b}\dist_{G'}(u,P) +(2t+1) \dist_{G'}(u,P)}{ \dist_{G'}(u,P)}\nonumber\\
&=\frac{4\alpha k}{b} + 2t+1.
\label{ineq:apxNE_case2}
\end{align}

In case $u\in P\setminus C_v$, consider the vertex $y\in C_v$ to which $u$ buys its only edge. 
Since the edge is the only one owned by $u$, the agent's edge cost in $G$ is $\alpha\cdot \norm{u,y}$. 
Next we evaluate the distance cost. 
  Every path from $u$ to a node $x\in P$ goes over $y$ in $G$. 
  If $x\in B_v$, by the triangle inequality and since $G$ is a $t$-spanner on $C_v$, \begin{align*}\dist_G(u,x)&=\norm{u,y}+\dist_G(y,x)\leq \norm{u,y}+ t\cdot \norm{y,x} \\ &\le \norm{u,y}+ t\cdot (\norm{y,v}+\norm{v,x})\leq \norm{u,y}+3t\frac{\maxDist}{b}.\end{align*} 
  
  If $x\notin B_v$, the shortest $u$-$x$ path goes over the $t$-spanner and contains at most two edges to nodes outside of $C_v$. 
  Since the length of any shortest path in $C_v$ is at most $2\cdot \frac{2\maxDist}{b}\cdot t$, then $\dist_G(u,x)\leq 2\maxDist + 4t\cdot\frac{\maxDist}{b}$.  
  So, we can evaluate the maximum improvement for $u$ as follows:
  \begin{align*}
   &\frac{\cost(u,G)}{\cost(u,G')} 
  \leq \frac{\alpha\cdot \norm{u,y} + \dist_G(u,B_v) + \dist_G(u,P\setminus B_v)}{\alpha\cdot \norm{u,S'_u} + \dist_{G'}(u,P)}\\
  &\leq\frac{\alpha\cdot \norm{u,y} + |B_v|\left(\norm{u,y}+\frac{3t\maxDist}{b}\right) + |P\setminus B_v|\left(2\maxDist+\frac{4t\maxDist}{b}\right)}{\dist_{G'}(u,B_v)}
  \end{align*}
	  
  In the denominator we observe that  $\dist_{G'}(u,B_v)\geq |B_v|\cdot\frac{1}{2}(\norm{u,y}+\frac{\maxDist}{b})$. 
  Indeed, for any $x\in B_v\subseteq C_v$, $d_{G'}(u,x)\geq \norm{u,x}\geq \norm{u,y}$ because $y$ is the closest node to $u$ in $C_v$. 
  Also, $u$ has distance at least $\frac{\maxDist}{b}$ to all nodes in $B_v$, since $u \notin C_v$. This yields that the ratio between   $\cost(u,G)$ and $\cost(u,G')$ is at most
   \begin{align}
  &\frac{\alpha\cdot \norm{u,y} + |B_v|\left(\norm{u,y}+\frac{3t\maxDist}{b}\right) + |P\setminus B_v|\left(2\maxDist+\frac{4t\maxDist}{b}\right)}{\frac{1}{2}\norm{u,y}\cdot |B_v| + \frac{\maxDist}{2b}\cdot |B_v|} 
  \label{PoA:metric:case:22:2}
  \end{align}
   
   Since $|B_v| > n-c$, we obtain $ \frac{(\alpha + |B_v|)\norm{u,y}}{\frac{1}{2}\norm{u,y}|B_v|} \le \frac{2\alpha}{n-c}+2$. 
   For the remaining part we get:
   \begin{align}
		\frac{\maxDist\left(|P\setminus B_v|(2+\frac{4t}{b})+|B_v|\frac{3t}{b}\right)}{ \frac{\maxDist}{2b}\cdot|B_v|}
		&\le \frac{\left(2+\frac{4t}{b}\right) |P\setminus B_v|}{\frac{1}{2b}\cdot|B_v|}+6t\nonumber\\
		&\le\frac{4c(b + 2t)}{n-c}+6t.\label{PoA:metric:case22:4}
	\end{align}
	Combining the two inequalities above we obtain\footnote{The implication follows from the following observation: for any $a, b, c, d\in \mathbb{R}_{> 0}$, $\frac{a+b}{c+d}\leq \max\big\lbrace\frac{a}{c}, \frac{b}{d}\big\rbrace$ } an upper bound for Inequality~\eqref{PoA:metric:case:22:2} equal to $\max\big\lbrace \frac{2\alpha}{n-c}+2, \frac{4c(b + 2t)}{n-c}+6t\big\rbrace$. 
	Together with Inequality \eqref{ineq:apxNE_case1} and Inequality \eqref{ineq:apxNE_case2} we obtain the upper bound of $$\max\left\{ \frac{k b}{c}\alpha +t, \frac{4 k}{b}\alpha  + 2t+1, \frac{2\alpha}{n-c}+2, \frac{4c(b + 2t)}{n-c}+6t\right\}. \qedhere$$
\end{proof}
\noindent The next result shows that the proof of Theorem~\ref{thm:apx_NE} also provides an upper bound for the social cost of the network computed by Algorithm~\ref{alg:apxNE}. This yields a $\apxNEopt{\beta}{\beta}$.

\begin{restatable}{theorem}{thmApxGraph}\label{thm:apx_graph}
	Let $b\in \mathbb{R}_{\geq 1}, k, t\in \mathbb{R}_{>1}, c\in \{0,\ldots, n-1\}$.
	Algorithm~\ref{alg:apxNE} computes  a $\apxNEopt{\beta}{\beta}$ with $$\beta = \max\left\{ \frac{k b}{c}\alpha +t, \frac{4 k}{b}\alpha  + 2t+1, \frac{2\alpha}{n-c}+2, \frac{4c(b + 2t)}{n-c}+6t\right\}.$$
\end{restatable}
\begin{proof}
	Let $G=G(\mathbf{s})$ be the network corresponding to the strategy profile computed by Algorithm~\ref{alg:apxNE}.
	By Theorem~\ref{thm:apx_NE}, $G$ is a $\beta$-NE.
	To complete the proof we need to show that the social cost of $G$ is at most $\beta$ times the social cost of $\opt{P}$.
	We evaluate the social cost of $G$ with respect to the social cost of the social optimum network $G^*$.
	Clearly, $\frac{\SC(G)}{\SC(G^*)}=\frac{\sum_{v\in P}\cost(v, G)}{\sum_{v\in P}\cost(v, G^*)}\leq \max\limits_{v\in P}\frac{\cost(v,G)}{\cost(v,G^*)}$.
	Therefore, we can repeat the analysis from the proof of Theorem~\ref{thm:apx_NE}.
	All upper bounds for the worst-case agent's improvement from the proof of Theorem~\ref{thm:apx_NE} provide the upper bounds for the ratio between the cost of the agent in $G$ and $G^*$. It holds because we did not assume for $G'$ that strategies of all other agents are the same as in $G$, i.e., we can replace $G'$ with $G^*$ in all inequalities.
	Hence, $\max_{v\in P}\frac{\cost(v,G)}{\cost(v,G^*)}\leq \beta$, and the statement follows.
\end{proof}

\begin{restatable}{corollary}{cor:big_n_ane}\label{cor:big_n_ane}
	Let $\alpha\leq n^x$ for some $x\in\mathbb{R}_{>0}$. Then we can construct in $\mathcal{O}(n^2)$ time a \mbox{\apxNEopt{$\beta$}{$\beta$}} with $\beta\in\mathcal{O}\left(\alpha^{1-\frac{1}{2x}}+1\right)$ for $x\ge 1$ and $\beta\in\mathcal{O}\left(\alpha^{\frac{3x-1}{4x}}+1\right)$ for $0< x < 1$. 
\end{restatable}
\begin{proof}
	Consider a \apxNEopt{$\beta$}{$\beta$} constructed by Algorithm~\ref{alg:apxNE} with parameters $b\leq \sqrt{2(n-1)}$ and $c=\frac{b^2}{2}$, some real number $t>1$ and $k\in \mathcal{O}\left((t-1)^{1-2d}\right)$. A $k$-degree $t$-spanner can be constructed in $\mathcal{O}(n\log n)$ time (\cite{narasimhan2007geometric}, Section 10.1).
	Hence, our algorithm outputs the $(\beta,\beta)$-NE in $\mathcal{O}(n^2)$ time with $$\beta =\max\left\{ \frac{4 k}{b}\alpha  + 2t+1, \frac{2\alpha}{n-\frac{b^2}{2}}+2, \frac{2b^2(b + 2t)}{n-\frac{b^2}{2}}+6t\right\}.$$
	
	We make a case distinction for $x$ when choosing $b$.
	If $x\geq 1$, we choose $b=\alpha^\frac{1}{2x}$.
	Then we get $\frac{4 k}{b}\alpha  + 2t+1=4k\alpha^{1-\frac{1}{2x}} + 2t+1\in\mathcal{O}(\alpha^{1-\frac{1}{2x}})$.
	Since $\alpha\leq n^x$, we have that $n-\frac{1}{2}b^2\ge \alpha^\frac{1}{x}-\frac{1}{2}\alpha^\frac{1}{x}\geq\frac{1}{2}\alpha^{\frac{1}{x}}$.
	Therefore, $\max\big\{\frac{2\alpha}{n-\frac{1}{2}b^2}+2, \frac{2b^2(b + 2t)}{n-\frac{1}{2}b^2}+6t\big\}\in \mathcal{O}\left(\max\{\alpha^{1-\frac{1}{x}}, \alpha^{\frac{1}{2x}}\}\right)$.
	Since we assume $x\geq 1$, we get $$\beta\in \mathcal{O}\left(\max\{\alpha^{1-\frac{1}{x}}, \alpha^{\frac{1}{2x}}, \alpha^{1-\frac{1}{2x}}\}\right) = \mathcal{O}\left(\alpha^{1-\frac{1}{2x}}\right).$$

	In case $0< x<1$, we choose $b=\alpha^\frac{x+1}{4x}$.
	Then we get $n-\frac{1}{2}b^2 \ge \alpha^\frac{1}{x}- \frac{1}{2}\alpha^\frac{x+1}{2x} \geq \alpha^\frac{1}{x}\left(1-\frac{1}{2}\alpha^\frac{x-1}{2x}\right)\geq \frac{1}{2}\alpha^\frac{1}{x}$ since $\alpha^\frac{x-1}{2x}<1$ for $0<x<1<\alpha$.
	Hence, we get for the approximation factor $\beta$ that $$\beta \in \mathcal{O}\left(\max\left\{\alpha^{1-\frac{x+1}{4x}}, \alpha^{1-\frac{1}{x}}, \alpha^{\frac{3(x+1)}{4x}-\frac{1}{x}} \right\}+1\right) = \mathcal{O}\left(\alpha^{\frac{3x-1}{4x}}+1 \right). \qedhere$$
\end{proof}

\noindent Corollary~\ref{cor:big_n_ane} claims that for $1\leq \alpha\leq \sqrt[3]{n}$, there is a $(\beta, \beta)$-NE with constant $\beta$, while for the other values of $\alpha$, the approximation is better than the one obtained for a clique (see Theorem~\ref{thm:apxNE_clique}).
However, when $x$ tends to infinity, the value of $\beta$ approaches $\alpha$.
For this case we show that a minimum spanning tree provides a better approximation.

\begin{restatable}{theorem}{thmMstAne}
	\label{thm:mst-ane}
	Any minimum spanning tree $MST(P)$ on a set of points $P$ is a $(n-1,n-1)$-network.
\end{restatable}
\begin{proof}
	Let $v,w \in P$, and let $p = (v =  v_1, \dots, v_k =  w)$ be the path connecting them in the minimum spanning tree.
	Since $MST(P)$ is a minimum spanning tree, $\norm{v,w} \ge \norm{v_i, v_{i+1}}$ for any $i = 1, \dots, k-1$.
	Thus, we get
	\begin{align*}
			\dist_{MST(P)}(v,w) &= \sum_{i=1}^{k-1} \norm{v_i, v_{i+1}} \le \sum_{i=1}^{k-1} \norm{v,w} \\
			&= (k-1)\norm{v,w} \le (n-1)\norm{v,w}
	\end{align*}
	and therefore $MST(P)$ is a $(n-1)$-spanner.

	Let $S_v$ be the strategy of agent $v$ in $T$.
	Consider the network $G'=(P,E')$, where $v$ plays a better strategy $S_v'$, i.e., any strategy that decreases agent $v$'s cost.
	The edge cost for $v$ in $G'$ is at least the same as in $MST(P)$ since no agent can delete any edge without buying new ones that are as expensive, due to $MST(P)$ being a minimum spanning tree.
	Thus, we get
	\[
		\frac{\cost(v,MST(P))}{ \cost(v,G')}\le \frac{\alpha\cdot \norm{v,S_v} + (n-1) \norm{v,P} }{\alpha\cdot \norm{v,S_v} + \norm{v,P}} \le (n-1) .
	\]
	Finally, if $G^*=(P,E^*)$ is a social optimum, then the total edge cost in $G^*$ is at least $\alpha\sum_{uv\in MST(P)}\norm{u,v}$, while the distance cost is at least $\sum_{v\in P}\norm{v,P}$.
	Since $MST(P)$ is a $(n-1)$-spanner, we get  \sloppy$\SC(MST(P))\leq (n-1)\cdot\SC(G^*)$, analogously to the first case.
	Hence, $MST(P)$ is a $\apxNEopt{n-1}{n-1}$.
\end{proof}

\noindent Finally, we show that using the better of the networks obtained by Algorithm~\ref{alg:apxNE} and the MST yields a $(\mathcal{O}(\alpha^\frac{2}{3}),\mathcal{O}(\alpha^\frac{2}{3}))$-network. See Figure~\ref{fig:apx_plot} for an illustration.
\begin{restatable}{corollary}{cor:apx_NE_all_alpha}
	A \apxNEopt{$\beta$}{$\beta$} with $\beta\in\mathcal{O}(\alpha^\frac{2}{3})$ can be constructed in $\mathcal{O}(n^2)$ time.
\end{restatable}
\begin{proof}
	Let $x\in\mathbb{R}^+$ such that $\alpha=n^x$. Then $\alpha^\frac{1}{x}=n$. Applying Theorem \ref{thm:mst-ane} for $x\ge\frac{3}{2}$ and Corollary \ref{cor:big_n_ane} for $x<\frac{3}{2}$ yields the result.
\end{proof}

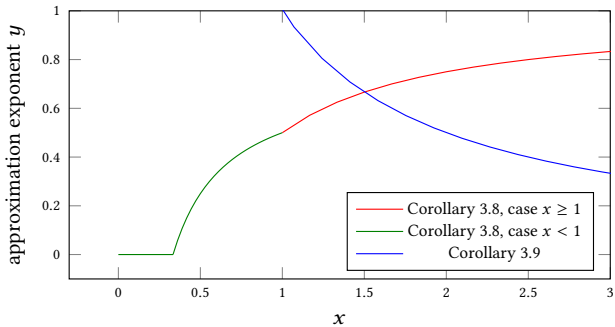
\begin{figure}[ht]
\centering
\begin{minipage}{\linewidth} 		
		\resizebox {\textwidth}{!}{
			\begin{tikzpicture}
				\begin{axis}[
            width=\textwidth,
            height=5cm,
					legend pos = south east ,
					legend style={nodes={scale=0.7, transform shape}}, 
					xlabel= {$x$},
					ylabel= {approximation exponent $y$},
					label style={font=\small},
					tick label style={font=\tiny},
					ytick = {0,0.2,0.4,0.6,0.8,1},
					ymax = 1,
					xmax = 3
					]
					
					\addplot[red, domain = 1:5]{1-(1/(2*x))};
					\addlegendentry{Corollary~\ref{cor:big_n_ane}, case $x\geq 1$};

					\addplot[green!50!black, domain = 0.3333:1]{(3*x-1)/(4*x))};
					\addlegendentry{Corollary~\ref{cor:big_n_ane}, case $x<1$};
					
					\addplot[blue, domain = 0.9:5]{1/x};
					\addlegendentry{Corollary~\ref{thm:mst-ane}};
					\addplot[green!50!black, domain = 0:0.3333]{0};
				\end{axis}
			\end{tikzpicture}
		}
		\end{minipage}
		\caption{Approximation factor $\beta$ obtained by Corollary \ref{cor:big_n_ane} and Theorem \ref{thm:mst-ane} dependent on the relation between $\alpha$ and $n$. For $x\in \mathbb{R}_{>0}$ such that $\alpha=n^x$, we can construct a \apxNEopt{$\beta$}{$\beta$} with $\beta\in\mathcal{O}(\alpha^y)$.}
	\label{fig:apx_plot}
	\end{figure}%

\noindent We now use Algorithm~\ref{alg:apxNE} to obtain a \apxNEopt{$1+\varepsilon$}{$1+\varepsilon$} for $\varepsilon>0$ if $P$ is chosen uniformly at random from the unit square.  We partition the unit square into four quadrants $C=\{a,b,c,d\}$ each containing a length-$\frac{1}{4}$-square from $C'=\{a', b', c', d'\}$, see Figure~\ref{fig:ane_beta_random_partition}. 
\begin{figure}[h]
	\centering
		\includegraphics[width=0.45\linewidth]{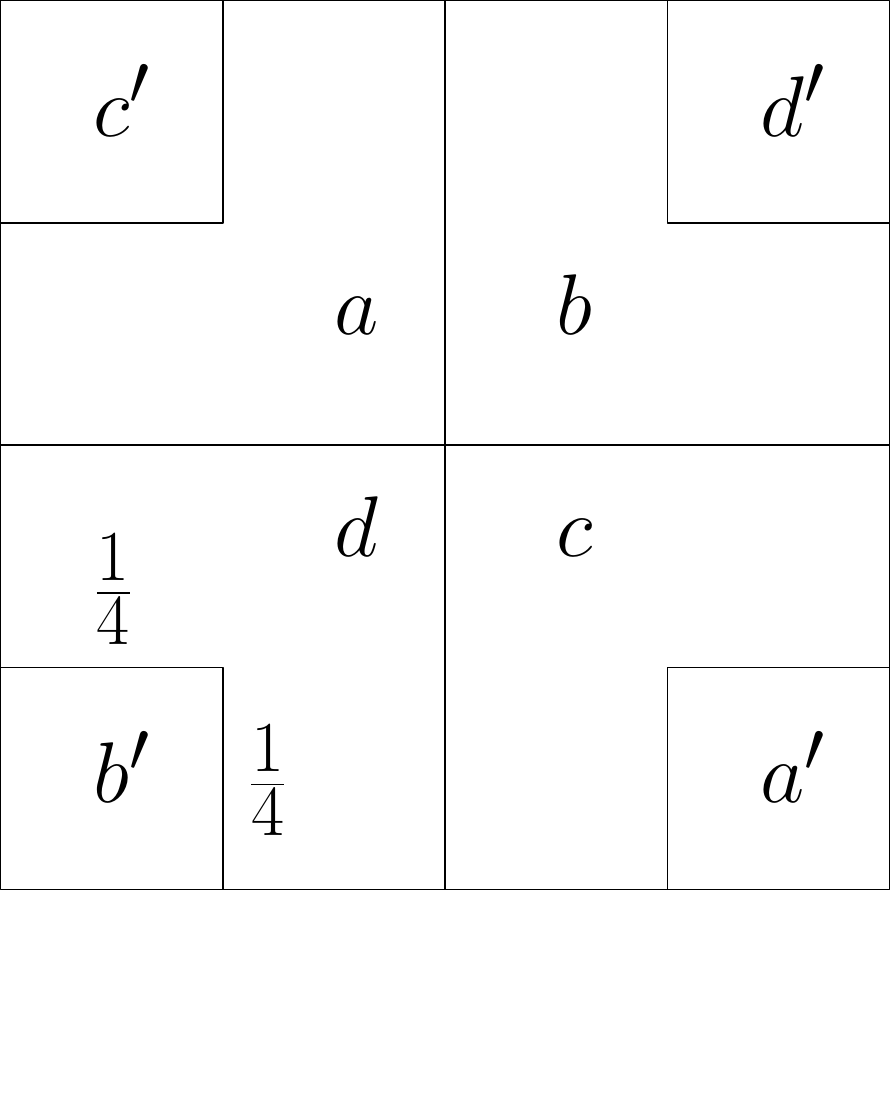}
\caption{Illustration of the partition of $[0,1]^2$ used in Lemma~\ref{lemma:random_chernoff} and Theorem~\ref{thm:apxNE_random}.}
	\label{fig:ane_beta_random_partition}
\end{figure}

The following Lemma shows that in this case with high probability the second case of Algorithm~\ref{alg:apxNE} occurs.

\begin{restatable}{lemma}{lemmaRandomChernoff} \label{lemma:random_chernoff}
	Let $P_n = \lbrace v_1, \dots, v_n\rbrace \subseteq [0,1]^2$ be  $n$ points chosen uniformly at random.
	Then $$
	\mathbb{P}\left(\bigcup_{\rho' \in C'} \left\lbrace\lvert  P_n \cap \rho'\rvert < (1-\delta)\frac{n}{16}\right\rbrace \right) \le 4	\exp\left(- \frac{\delta^2n}{32}\right).$$ 
\end{restatable}
\begin{proof}
	By the Union-Bound, showing that for all $\rho' \in C'$  $$\mathbb{P}\left(\lvert( P_n \cap \rho'\rvert < (1-\delta)\frac{n}{16} \right) \le \exp \left(- \frac{\delta^2n}{32} \right)$$ suffices to prove the statement.
	
	Let $X_i = 1$ if $v_i\in \rho'$ and $X_i = 0$, otherwise, for $1\leq i \leq n$.
	Obviously, $\sum_{i=1}^n X_i = \lvert P_n \cap \rho'\rvert$. Note that the area of each $\rho'\in C'$ is $\frac{1}{16}$, and therefore $\mathbb E\big[\sum_{i=1}^n X_i\big] = \frac{n}{16}$, by linearity of expectation. Thus, by Chernoff's inequality we get
	\[
	\mathbb{P}\left(\lvert P_n \cap \rho'\rvert < (1-\delta)\frac{n}{16} \right) \le \exp\left(- \frac{\delta^2n}{32}\right) . \qedhere
	\]
\end{proof}

\noindent By Theorem 10.1.3 from \cite{narasimhan2007geometric}, for any $\varepsilon> 0$, there is a $(1+\varepsilon)$-spanner with maximum degree only depending on $\varepsilon$. We use this construction to provide a $(1+\varepsilon, 1+\varepsilon)$-NE. 

\begin{restatable}{theorem}{thmApxNERandom}\label{thm:apxNE_random}
	Let $\varepsilon > 0$ and $P_n =\lbrace v_1, \dots, v_n\rbrace\subseteq [0,1]^2$ be a set of $n$ points chosen uniformly at random. 
Then if $\alpha_n \in o(n)$, there exists a $\apxNEopt{1+\varepsilon}{1+\varepsilon}$ for $P_n$, for any $\alpha < \alpha_n$ asymptotically almost surely.
\end{restatable}
\begin{proof} 
	Consider a network computed by Algorithm \ref{alg:apxNE} with parameters $b = 4$, $c = 2k_{\varepsilon}b\frac{\alpha }{\varepsilon}$, and $k_{\varepsilon}$-degree $\left(1+\frac{\varepsilon}{2}\right)$-spanner as constructed in \cite{narasimhan2007geometric}. 
	 
	Because $\alpha_n \in o(n)$, we can assume that $\alpha_n \le \frac{\varepsilon}{16 k_{\varepsilon}}\frac{n}{16}$, because it holds for all but finitely many $n$. 
	Since $$
	c = 2k_{\varepsilon}b\frac{\alpha }{\varepsilon} \le 2k_{\varepsilon}b\frac{1}{\varepsilon}\cdot\frac{\varepsilon}{16 k_{\varepsilon}}\cdot\frac{n}{16} = \left(1-\frac{1}{2}\right)\frac{n}{16},$$ then, by Lemma~\ref{lemma:random_chernoff}, each $\rho' \in C'$ has at least $c$ points with probability $1-4\exp(- \frac{n}{128})$. 
	Thus, for any point which is part of a quadrant $\rho\in C$, there are at least $c$ points in $\rho'$ within the distance at least $\frac{\maxDist}{b} = \frac{1}{4}$.
	Because we have a $k_{\varepsilon}$-degree $\left(1 + \frac{\varepsilon}{2}\right)$-spanner, by Theorem~\ref{thm:apx_NE}  and Theorem~\ref{thm:apx_graph}, we get, with probability $1-4\exp(- \frac{n}{128})$,  a \apxNEopt{$\beta$}{$\beta$} with 
	$$
	\beta = \frac{k_{\varepsilon}b}{c}\alpha+ 1 + \frac{\varepsilon}{2} =k_{\varepsilon}b\frac{\varepsilon}{2k_{\varepsilon}b\alpha}\alpha+ 1 + \frac{\varepsilon}{2} = 1 + \varepsilon
	. \qedhere$$
\end{proof}

\noindent Finally, we study integer grids in $\Rd$ and show that Algorithm~\ref{alg:apxNE} computes a $\apxNEopt{2d}{ 2d}$, if the grid itself is selected as spanner. 
\begin{restatable}{theorem}{theoGrids} \label{theo:d_dim_grids}
	Let $(b_1,\cdots,b_d) \in \mathbb{N}^d$ and $B = [0,b_1] \times \cdots \times [0,b_d]$ the corresponding hyperrectangle. Let $P \coloneqq \mathbb{Z}^d \cap B$. Then, there exists a $\apxNEopt{2d}{2d}$  for the nodes in $P$.
\end{restatable}
\begin{proof}
We now construct such a \apxNEopt{$2d$}{$2d$} for $P$. Let $N$ be the set of all nearest neighbor edges along the grid and $G=(P,N)$ the corresponding network. This corresponds to choosing $c = 0$ in Algorithm \ref{alg:apxNE} and letting the algorithm choose the grid as a spanner. In order to obtain better bounds, we  redo the analysis for this case.  Since $G$ is bipartite, we can assign the edges, such that one part $L$ of the bipartition buys all their edges to their respective $2d$ neighbors in partition $R$.  
First we prove that $G$ is a $\sqrt d$-spanner. Let $p,q \in P$, then using the Cauchy-Schwarz inequality we obtain  $$\left(\frac{d_G(p,q)}{\norm{p,q}}\right)^2 = \left(\frac{\sum\limits_{i=1}^{d} \lvert p_i - q_i \rvert}{\sqrt{\sum\limits_{i=1}^{d} ( p_i - q_i )^2}}\right)^2 \le\sum\limits_{i=1}^{d}\frac{\lvert p_i - q_i \rvert^2}{ ( p_i - q_i )^2}  = d$$ and thus $\frac{d_G(p,q)}{\norm{p,q}} \le \sqrt d$.
Now consider any agent $p \in L$. Since $p$ buys all edges to her neighbors in the grid, she has edge cost of at most $\alpha 2 d$. In every improving move, $p$ must keep at least one edge, since otherwise the network would get disconnected. Thus we get an approximation factor $\beta \le \frac{\alpha 2 d + \sqrt{d}\norm{p, P}}{\alpha + \norm{p,P}} \le 2d $. 

If $p \in R$, the agent does not buy any edges and we get $\beta \le \frac{ \sqrt{d}\norm{p, P}}{\norm{p,P}} \le \sqrt{d}$. Thus,  we have a $2d$-NE. 

Analogously, for the social optimum approximation we get  
$\alpha (n-1) + \sum_{p\in P} \norm{p, P}$ as a trivial lower bound for the social cost of the optimum and thus, $\frac{\alpha dn +\sqrt{d}\sum_{p\in P}\norm{p, P} }{\alpha(n-1) + \sum_{p\in P}\norm{p, P}} \le 2d.$ 
\end{proof}

\section{Price of Anarchy and Price of Stability}\label{sec:PoA}
In this section we will provide lower bounds on the Price of Anarchy and on the Price of Stability in the Euclidean case.
It was shown that in the 1-norm space, the lower bound for the PoA approaches the upper bound of $\frac{\alpha+2}{2}$ when the dimension $d$ tends to infinity \cite{bilo2019geometric}. 
We show that also in Euclidean space, there is an instance that asymptotically almost meets the upper bound.

\begin{restatable}{theorem}{theoremDInfty}
The Price of Anarchy in the \game \xspace is at least $\min\big\{\frac{\alpha + 1}{\sqrt 2},\frac{\alpha^2 + 2\alpha + 2}{2\alpha + 2}\big\}$ as $d\to \infty$. \label{theorem:d_to_infty}
\end{restatable}
\begin{proof}
To provide the lower bound for the PoA we consider the following set of points.
	Let $x > 0$, and $n = 2d$. We define a set of $n$ points $P = \lbrace m, u\rbrace \cup T$, where $m \coloneqq (0, \dots, 0) \in \mathbb{R}^d$ is the central point, $u \coloneqq (0, \dots, 0, x) \in \mathbb{R}^d$, and $T \coloneqq \lbrace (\delta_{i,j})_{j = 1}^d,  (-\delta_{i,j})_{j = 1}^d \mid i \in \lbrace 1, \dots, d-1\rbrace\rbrace$, with $\delta_{i,j}=1$ if $i=j$, and $\delta_{i,j}=0$, otherwise. For an illustration of $P$, refer to Figure~\ref{fig:lower_bound_d_infty}.
	\begin{figure}[h]
	\centering
	\includegraphics[scale=0.8]{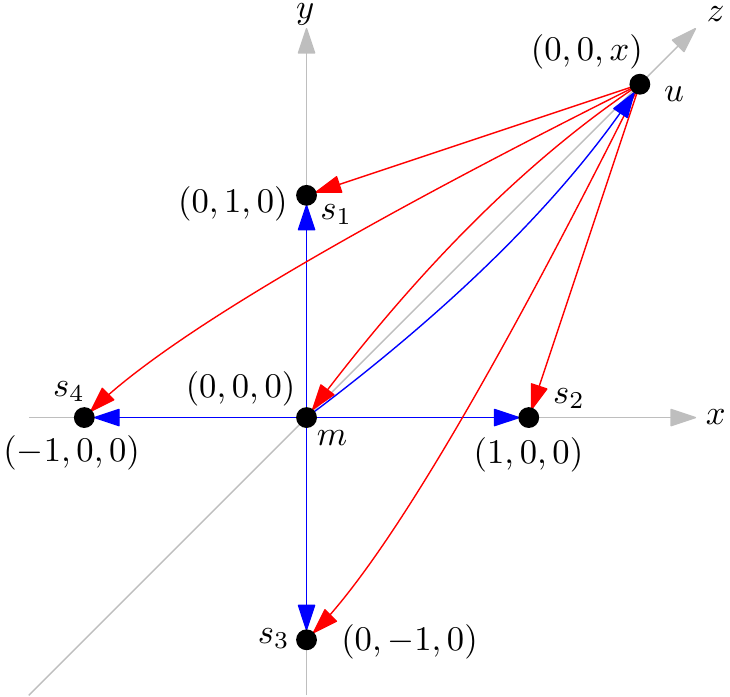}
	\caption{Illustration of the construction in the proof of Theorem~\ref{theorem:d_to_infty}.}
	\label{fig:lower_bound_d_infty}
\end{figure}

	Consider a star $S_n(u)$ centered at $u$ and a star $S_n(m)$ centered at $m$. 
First, we show that $S_n(u)$ is a NE, when all edges are owned by the central agent $u$. 
	Since $S_n(u)$ is a star, $u$ can neither buy nor sell edges to improve her strategy.
	The agent $m$ can change her strategy only by buying edges to the agents in $T$. 
	Note that buying an edge to a node in $T$ does not change the distance to any other node, hence we only have to show that it is not beneficial to buy a single edge, say an edge $mt$, where $t\in T$.
	Since $\norm{m,t}=1$, and the distance $d_{S_n(u)}(m,t)=\norm{m,u} + \norm{u,t}=x+\sqrt{1+x^2}$, the cost of $m$ after buying the edge changes by $\alpha + 1 - (\sqrt{1+x^2}+x)$.
	Thus, $m$ has no improving move if $x \le \frac{\alpha^2 + 2\alpha}{2\alpha + 2}$ holds.
	
	Analogously, any agent $t\in T$ cannot improve her strategy by buying the edge $tm$ if the inequality above holds.
	It remains to examine the case when $t$ buys an edge to another node $w\in T$.
	The length of the edge is at least $\sqrt{2}$ and the distance between nodes in $S_n(u)$ is $2\sqrt{x^2+1}$. 
	Thus, the cost changes by at least \sloppy${\sqrt{2}\alpha+\sqrt{2}-2\sqrt{x^2+1}}$. 
	Therefore the edge is not bought if $x\le \sqrt{\frac{1}{2}(\alpha^2 + 2\alpha - 1)}$.
	
	For $\alpha \geq \sqrt{1+\sqrt{2}}-1$, we have $\sqrt{\frac{1}{2}(\alpha^2 + 2\alpha - 1)} \geq \frac{\alpha^2 + 2\alpha}{2\alpha + 2}$.
	In that case, with ${x=\frac{\alpha^2 + 2\alpha}{2\alpha + 2}} $ the above inequalities that are necessary for $S_n(u)$ be stable are satisfied, and the corresponding star $S_n(u)$ is in NE. 
	In case $\alpha < \sqrt{1+\sqrt{2}}-1$,  $S_n(u)$ is a NE for $x=\sqrt{\frac{1}{2}(\alpha^2 + 2\alpha - 1)}$.

We proved that $S_n(u)$ is in NE. Next, we evaluate its  social cost.
	The edge cost of the star $S_n(u)$ equals  $\alpha \norm{u,T} + \alpha \norm{m,u}=(n-2)\alpha \sqrt{1+x^2} + \alpha x$, while the distance cost is $(2n-2)x+  (2n^2-6n+4)\sqrt{1+x^2}$.
	
	Since our final aim is to provide a lower bound for the PoA, we need an upper bound for the social cost of the optimum network. 
	For this, we consider a star $S_n(m)$ centered at the node $m$.\footnote{It is easy to verify that the star $S_n(m)$ is the social optimum for $\alpha\geq 1$.} 
	The social cost of $S_n(m)$ is $(n-2)\alpha + \alpha x + (2n-2)x+  (2n^2-6n+4)$.
	Then we get:
\begin{align*}	
	\frac{\cost(S_n(u))}{\cost(S_n(m))} &= \frac{\sqrt{1+x^2}((n-2)\alpha +2n^2-6n+4) +(2n-2+\alpha)x}{(n-2)\alpha +2n^2-6n+4+ (2n-2+\alpha)x}\\
	&\xrightarrow[n\to\infty]{}\sqrt{1+x^2}.
\end{align*}
	In case $\alpha \geq \sqrt{1+\sqrt{2}}-1$, it was shown above that $S_n\left(u\right)$ is a NE for ${x=\frac{\alpha^2 + 2\alpha}{2\alpha + 2}} $, the last coordinate of the node $u$.
	Thus, for sufficiently large $n$, the ratio between the social cost of the NE and the optimum approaches $\frac{\alpha^2 + 2\alpha+2}{2\alpha + 2}$.
	In case $\alpha < \sqrt{1+\sqrt{2}}-1$, $S_n(u)$ is a NE for $x = \sqrt{\frac{1}{2}(\alpha^2 + 2\alpha - 1)}$.
	Therefore, the ratio between $\cost(S_n(u))$ and $\cost(S_n(m))$ tends to $\frac{\alpha+1}{\sqrt{2}}$ as $n$ tends to infinity.
	This completes the proof.
	\end{proof}

\noindent Next, we now show that the PoA is super-constant in $\alpha$, even when the underlying space is the $\mathbb{R}^1$. For this we significantly improve the analysis of a construction of Bilò et al. (\cite{bilo2019geometric}, Theorem 3.27). We start with technical lemma.

\begin{lemma}\label{lemma:weird_sum}
	Let $\alpha>0$ and $n\in\mathbb{N}_+$. Then
	\begin{align*}
		&2n+\sum_{i=1}^{n-1}\frac{4}{\alpha}\left(1+\frac{2}{\alpha}\right)^{i-1}(i+1)(n-i)\\
		&=(\alpha n-\alpha^2)\left(1+\frac{2}{\alpha}\right)^n+\alpha^2+\alpha n.
	\end{align*}
\end{lemma}
\begin{proof}
	We proof the statement by induction over $n$. For $n=1$ the statement clearly holds.
	
	Now let $n>1$ such that the statement holds. We show that the statement also holds for $n+1$. We have
	\begin{align*}
		&2(n+1)+\sum_{i=1}^{n}\frac{4}{\alpha}\left(1+\frac{2}{\alpha}\right)^{i-1}(i+1)(n+1-i)\\
		=&2n+2+\sum_{i=1}^{n}\frac{4}{\alpha}\left(1+\frac{2}{\alpha}\right)^{i-1}(i+1)(n-i)+\sum_{i=1}^{n}\frac{4}{\alpha}\left(1+\frac{2}{\alpha}\right)^{i-1}(i+1)
	\end{align*}
	Applying the induction hypothesis, we have $$2n+\sum_{i=1}^{n-1}\frac{4}{\alpha}\left(1+\frac{2}{\alpha}\right)^{i-1}(i+1)(n-i)=(\alpha n-\alpha^2)\left(1+\frac{2}{\alpha}\right)^n+\alpha^2+\alpha n.$$ For the rest of the term we first split the sum to get 
	\begin{align*}
		&2 +\sum_{i=1}^{n}\frac{4}{\alpha}\left(1+\frac{2}{\alpha}\right)^{i-1}(i+1)\\
		=&2 +\sum_{i=0}^{n-1}\frac{4}{\alpha}\left(1+\frac{2}{\alpha}\right)^{i} \cdot i +2\sum_{i=0}^{n-1}\frac{4}{\alpha}\left(1+\frac{2}{\alpha}\right)^{i}.
    \end{align*}
    Now by applying theorems for geometric series this equals
    $$2 +\frac{4}{\alpha}\cdot\frac{(n-1)\left(1+\frac{2}{\alpha}\right)^{n+1}-n\left(1+\frac{2}{\alpha}\right)^{n}+1+\frac{2}{\alpha}}{\left(1+\frac{2}{\alpha}-1\right)^2}
		+\frac{8}{\alpha}\cdot\frac{\left(1+\frac{2}{\alpha}\right)^{n}-1}{1+\frac{2}{\alpha}-1} $$
    and we can simplify this further to
    \begin{align*}
		&2+\alpha\left((n-1)\left(1+\frac{2}{\alpha}\right)^{n+1}-n\left(1+\frac{2}{\alpha}\right)^{n}+1+\frac{2}{\alpha}\right)
		+4\left(1+\frac{2}{\alpha}\right)^{n}-4\\
		=&\alpha\left((n-1)\left(1+\frac{2}{\alpha}\right)^{n}+\frac{2(n-1)}{\alpha}\left(1+\frac{2}{\alpha}\right)^{n}-n\left(1+\frac{2}{\alpha}\right)^{n}+1+\frac{2}{\alpha}\right)\\
		&+4\left(1+\frac{2}{\alpha}\right)^{n}-2\\
		=&(\alpha (n-1)+ 2(n-1)-\alpha n+4)\left(1+\frac{2}{\alpha}\right)^{n}+\alpha+2-2\\
		=&(2n-\alpha+2)\left(1+\frac{2}{\alpha}\right)^{n}+\alpha.
	\end{align*}
	Together we yield
	\begin{align*}
		&(\alpha n-\alpha^2)\left(1+\frac{2}{\alpha}\right)^n+\alpha^2+\alpha n+(2n-\alpha+2)\left(1+\frac{2}{\alpha}\right)^{n}+\alpha\\
		=&(\alpha n -\alpha^2 +2n -\alpha +2)\left(1+\frac{2}{\alpha}\right)^n+\alpha^2+\alpha (n+1)\\
		=&(\alpha n+\alpha-\alpha^2)\left(1+\frac{2}{\alpha}\right)\left(1+\frac{2}{\alpha}\right)^{n}+\alpha^2+\alpha (n+1)\\
		=&(\alpha (n+1)-\alpha^2)\left(1+\frac{2}{\alpha}\right)^{n+1}+\alpha^2+\alpha (n+1).
	\end{align*}
	Therefore the statement also holds for $n+1$ and the lemma follows by induction.
\end{proof}

\begin{restatable}{theorem}{RLBPOA}\label{thm:R1_PoA_LB}
	The PoA in the $\mathbb{R}^1$ and thus in the $\mathbb{R}^d$ is lower bounded by $\frac{3}{5}\alpha^{\frac{2}{3}}\pm{o}(\alpha^\frac{2}{3})$.
\end{restatable}
\begin{proof}
	
	We construct a set of $n+1$ points $P=\{p_0,\dots,p_n\}$ in the $\mathbb{R}^1$ with coordinates $p_0=0$ and for $1\le i\le n\colon p_i=\left(1+\frac{2}{\alpha}\right)^{i-1}$. For an illustration of the construction, refer to Figure~\ref{fig:LB_PoA2}.
	\begin{figure}[h]
		\centering
		\includegraphics[width=\linewidth]{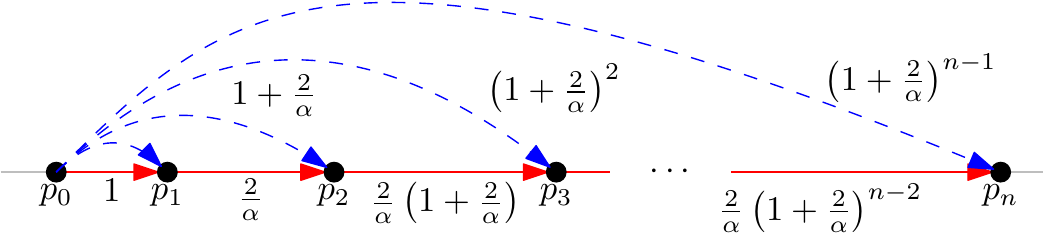}
		\caption{Illustration of the lower bound construction in the $\mathbb{R}^1$-GNCG. The blue dashed edges are in NE, the red solid edges are a social optimum.}
		\label{fig:LB_PoA2}
	\end{figure}
	We have $\norm{p_1,p_0}=1$ and for $2\le i \le n$ we have $\norm{p_i,p_{i-1}}=\frac{2}{\alpha}\cdot\left(1+\frac{2}{\alpha}\right)^{i-2}$. Let $\mathbf{s}=(\{p_1,\dots,p_n\},\varnothing,\dots,\varnothing)$ be the strategy-profile of the star with $p_0$ as the center node buying all edges. Let furthermore $\mathbf{s^*}=(\{p_1\},\dots,\{p_n\},\varnothing)$ be the strategy-profile where each point buys only the edge to the next point. Bilo et al. showed, that $G(\mathbf{s})$ is an NE and $G(\mathbf{s^*})$ is a social optimum and that the social costs of $G(\mathbf{s})$ are $\SC(G(\mathbf{s}))=\left(\frac{\alpha}{2}\left(1+\frac{2}{\alpha}\right)^n-\frac{\alpha}{2}\right)(2n+\alpha)=\alpha\left(\left(1+\frac{2}{\alpha}\right)^n-1\right)\left(n+\frac{\alpha}{2}\right)$ \cite{bilo2019geometric}. For the distance cost of $G(\mathbf{s^*})$ we count for each edge in how many shortest paths it contains. We get
	\begin{align*}
		2\cdot&\sum_{i=0}^{n-1}\norm{p_{i+1},p_i}(i+1)(n-i)\\
		=&2n+\sum_{i=1}^{n-1}\frac{2}{\alpha}\left(1+\frac{2}{\alpha}\right)^{i-1}(i+1)(n-i)\cdot 2\\
		=&\alpha\left((n-\alpha)\left(1+\frac{2}{\alpha}\right)^n+\alpha+n\right).
	\end{align*}
	The correctness of the last step can be verified by an induction over $n$. See Lemma~\ref{lemma:weird_sum} for details.
	Adding the buying cost yields $\SC(G(\mathbf{s^*}))=\alpha\left((n-\alpha)\left(1+\frac{2}{\alpha}\right)^n+\alpha+n+\left(1+\frac{2}{\alpha}\right)^{n-1}\right)$.
	
	We now bound the PoA as
	\begin{equation}
		PoA\ge\frac{\SC(\mathbf{s})}{\SC(\mathbf{s^*})}=\frac{\left(\left(1+\frac{2}{\alpha}\right)^n-1\right)\left(n+\frac{\alpha}{2}\right)}{(n-\alpha)\left(1+\frac{2}{\alpha}\right)^n+\alpha+n+\left(1+\frac{2}{\alpha}\right)^{n-1}}\label{poa}.
	\end{equation}
	Next we choose $n=\alpha^\frac{2}{3}$. Note that $n$ needs to be an integer. Since the statement is asymptotic, we can assume it without spoiling the result. With the Binomial Theorem, we obtain
	\begin{align*}
		\left(1+\frac{2}{n^\frac{3}{2}}\right)^n&=\sum_{i=0}^n \binom{n}{i}\frac{2^i}{n^\frac{3i}{2}}\\
		&=1+\frac{2n}{n^\frac{3}{2}}+\frac{4n(n-1)}{2n^3}+\frac{8n(n-1)(n-2)}{6n^\frac{9}{2}}+\sum_{i=4}^{n}\binom{n}{i}\frac{2^i}{n^\frac{3i}{2}}\\
		&=1+\frac{2}{n^\frac{1}{2}}+\frac{2}{n}+\frac{4}{3n^\frac{3}{2}}\pm\mathcal{O}\left(\frac{1}{n^2}\right)\\
		&=1+\frac{2}{\alpha^\frac{1}{3}}+\frac{2}{\alpha^\frac{2}{3}}+\frac{4}{3\alpha}\pm\mathcal{O}\left(\frac{1}{\alpha^\frac{4}{3}}\right).
	\end{align*}
	We observe that $\lim\limits_{n\rightarrow\infty}\left(1+\frac{2}{n^\frac{3}{2}}\right)^{n-1}=1$.
	Inserting into the denominator of \eqref{poa} yields
	\begin{align*}
		&(\alpha^\frac{2}{3}-\alpha)\left(1+\frac{2}{n^\frac{3}{2}}\right)^n+\alpha+\alpha^\frac{2}{3}+\left(1+\frac{2}{n^\frac{3}{2}}\right)^{n-1}\\
		&=(\alpha^\frac{2}{3}-\alpha)\left(1+\frac{2}{\alpha^\frac{1}{3}}+\frac{2}{\alpha^\frac{2}{3}}+\frac{4}{3\alpha}\pm\mathcal{O}\left(\frac{1}{\alpha^\frac{4}{3}}\right)\right)\\
		&+\alpha+\alpha^\frac{2}{3}+\left(1+\frac{2}{n^\frac{3}{2}}\right)^{n-1}\\
		&=\alpha^\frac{2}{3}-\alpha+2\alpha^\frac{1}{3}-2\alpha^\frac{2}{3}+2-2\alpha^\frac{1}{3}-\frac{4}{3}\pm\mathcal{O}\left(\frac{1}{\alpha^\frac{1}{3}}\right)\\
		&+\alpha+\alpha^\frac{2}{3}+\left(1+\frac{2}{n^\frac{3}{2}}\right)^{n-1}\\
		&=2-\frac{4}{3}\pm\mathcal{O}\left(\frac{1}{\alpha^\frac{1}{3}}\right)+\left(1+\frac{2}{n^\frac{3}{2}}\right)^{n-1}
		\xrightarrow[\alpha\to\infty]{}\frac{5}{3}.
	\end{align*}
	For the numerator of \eqref{poa}, we get
	\begin{align*}
	&\left(\left(1+\frac{2}{n^\frac{3}{2}}\right)^n-1\right)\left(\alpha^\frac{2}{3}+\frac{\alpha}{2}\right)\\
		&=\left(\frac{2}{\alpha^\frac{1}{3}}+\frac{2}{\alpha^\frac{2}{3}}+\frac{4}{3\alpha}\pm\mathcal{O}\left(\frac{1}{\alpha^\frac{4}{3}}\right)\right)\left(\alpha^\frac{2}{3}+\frac{\alpha}{2}\right)\\
		&=\alpha^\frac{2}{3}\pm\mathcal{O}\left(\alpha^\frac{1}{3}\right).
		\end{align*}
	Together we get $\frac{3}{5}\alpha^\frac{2}{3}\pm {o}(\alpha^\frac{2}{3})$ as a lower bound for the PoA.
\end{proof}

\noindent Finally, we show that the PoS is strictly larger than $1$. 
\begin{theorem}\label{thm:PoS}
The PoS in the $\Rd$ with $d\geq 2$ is greater than $1$ if $\alpha > 2$. 
\end{theorem}
\begin{proof}[Proofsketch]
Consider the construction from the proof of Theorem~\ref{thm:social_improv}.  After connecting the nodes within each cluster there are only two networks that connect the network; one with two length-1-edges and one with three length-1-edges. We set $n = 3 (\lceil \alpha\rceil -1)$ and observe that buying all three length-1-edges results in a social optimum network, as $2 \cdot 2 \cdot \left(\frac{n}{3}\right)^2 > \alpha +  2 \cdot \left(\frac{n}{3}\right)^2$. Selling a length-1-edge leads to a NE as $2 \cdot \frac{n}{3} < \alpha + 1\cdot \frac{n}{3}$ . This means that the social optimum network is not a NE, showing that PoS $ > 1$.
\end{proof}

\section{Outlook: Efficiency and Stability on a Host Network}\label{sec:GNCG}
We consider a more general model, i.e., the \emph{Generalized Network Creation Game (GNCG)} by Bilò et al.~\cite{bilo2019geometric}, where a complete host network $H = (V,E)$ with arbitrary edge weights $w\colon V \times V\rightarrow\mathbb{R}_+$ is given. The price of an edge $uv\in E$ then is $\alpha\cdot w(u,v)$. For an edge subset $E'\subseteq E(H)$, we denote $w(E')\coloneqq \sum_{uv\in E'}w(u,v)$. 
 
 Our main contribution of this section indicates that the geometric and non-geometric versions of the GNCG behave very similarly. Clearly, the hardness results carry over from the special case to the more general case. But we also extend the approximation results for stable and optimum networks to the GNCG, and we show that the PoA is linear in $\alpha$, as conjectured in \cite{bilo2019geometric}. Thus it matches the bounds for the metric version.

\subsection{Approximation}
Theorem~\ref{thm:apxNE_clique} and Theorem~\ref{thm:mst-ane} can be directly generalized to the case with arbitrary, even non-metric, edge weights if we consider a spanning sub-network $H'$ of the host network $H$ such that each edge in $H'$ participates in at least one shortest path.

\begin{corollary}
	Let $H=(V,E)$ be a host network in the GNCG. A spanning sub-network $H'=(V, E')$ of $H$, where $E'=\{(u,v) \mid w(u,v)=\dist_H(u,v)\}$, is a $\left(\alpha+1, \frac{\alpha}{2}+1\right)$-NE.
\end{corollary}

\begin{corollary}
Consider a host network $H=(V,E)$. Any minimum spanning tree $MST(H)$ of $H$ is a $\apxNEopt{n-1}{n-1}$.
\end{corollary}
\noindent The key idea of the two results above is to remove long edges from $H$ to obtain an approximation similar to the case with metric edge weights. 
We can extend this idea to make our Algorithm 1 work for the GNCG with arbitrary edge weights as follows. 

Consider a host network $H = (V,E)$. 
Update $H$ as follows: starting from the longest edge $uv\in E$, if $d_H(u,v)<w(u,v)$ remove $uv$ from $H$. 
Repeat the procedure until all edges are checked. 
Denote the final network as $H_M = (V,E_M)$. 
Note that $H_M$ is connected and has metric edge weights, i.e., for any edge $uv \in E_M$ we have $w(u,v)=d_{H_M}(u,v)$. 
Now we can apply Algorithm~\ref{alg:apxNE} to $H_M$ with the only modification on step 7: connect node $u$ with the closest node $u'\in C_v$ by the shortest path $\pi_{H_M}(u,u')$. We assume $w_{\max}$ is the length of the longest shortest path in $H_M$.
\begin{corollary}
Consider a host network $H=(V,E)$, and let $b\in \mathbb{R}_{\geq 1}, k, t\in \mathbb{R}_{>1}, c\in \{0,\ldots, n-1\}$ such that a $k$-distributable $t$-spanner\footnote{See Footnote~\ref{fn:k_distr_spanner} for a definition of a $k$-distributable $t$-spanner.} exists for $H_M$.
	Algorithm~\ref{alg:apxNE} computes  a $\apxNEopt{\beta}{\beta}$ for $H$ with $$\beta = \max\left\{ \frac{k b}{c}\alpha +t, \frac{4 k}{b}\alpha  + 2t+1, \frac{2\alpha}{n-c}+2, \frac{4c(b + 2t)}{n-c}+6t\right\}.$$
\end{corollary} 
\noindent However, the above statement relies on the existence of a $k$-dist\-ri\-but\-able $t$-spanner for an incomplete weighted host network, which is, to the best of our knowledge, still open. Hence, finding an efficient algorithm that computes a $k$-distributable $t$-spanner with low $k$ and $t$ would enable Algorithm $1$ to obtain $(\beta,\beta)$-networks with low $\beta$ for the GNCG.

\subsection{Price of Anarchy} We show a $\mathcal{O}(\alpha)$ upper bound on the PoA. This asymptotically matches the $\Omega(\alpha)$ lower bound from Bilò et al.~\cite{bilo2019geometric}.

\begin{restatable}{theorem}{thmPoAGNCG}\label{thm:PoA_gncg}
In the GNCG the PoA is at most $2(\alpha+1)$.
\end{restatable}
\begin{proof}
Consider a host network $H=(V,E(H))$, a stable network $G=(V,E)$ and an optimum network $G^*=(V,E^*)$. 
Since every NE is a $(\alpha+1)$-spanner (by Lemma 2.2 in \cite{bilo2019geometric}) the distance cost of $G$ is
\begin{align}
		\sum_{u,v\in V}d_G(u,v)\leq(\alpha+1)\sum_{u,v\in V}d_H(u,v)\leq (\alpha+1)\sum_{u,v\in V}d_{G^*}(u,v).\label{poa:gncg:thirdsummand}
\end{align}
Now we evaluate the edge cost of $G$. 
We partition the edges in $E$ concerning the edges in the optimum and analyze each set's total cost separately. 
Let $B\coloneqq\bigcup_{uv\in E^*}\{\pi_G(u,v)\}$ be a set of edges in $G$ appearing in some shortest $u$-$v$ path in $G$ for each $uv\in E^*$, i.e., for every edge $uv$ in the social optimum, $B$ contains all edges from a shortest path between $u$ and $v$ in $G$. 
We denote the rest of the edges in $G$ as $R\coloneqq E\setminus B$.
	
Since $G$ is a $(\alpha+1)$-spanner, we can evaluate the cost of the edges in $B$ as follows:
\begin{align}
		\alpha\cdot w(B) & =\alpha\cdot w\left(\bigcup_{uv\in E^*} \pi_G(u,v)\right)
		 \leq\alpha\sum_{uv\in E^*}d_G(u,v)\nonumber \\
           &\leq \alpha (\alpha+1)\sum_{uv\in E^*}d_{H}(u,v)      
		             \leq \alpha (\alpha+1)w(E^*).\label{poa:gncg:firstsummand}
	\end{align}
	
	Next, we compute the cost of edges in $R$. 
	Consider an agent $u\in V$. 
	Let $R_u$ be a set of edges from $R$ that $u$ buys in $G$. 
	Since $G$ is a NE, deleting all edges from $R_u$ is not an improving move for $u$. 
	Hence, $$\alpha w(R_u) + \sum_{v\in V}d_{G}(u,v) \le \sum_{v\in V}d_{G-R_u}(u,v),$$ where $G-R_u$ is the network obtained after the deletion.
	Thus, $$\alpha w(R_u)\leq  \sum_{v\in V}d_{G-R_u}(u,v) - d_{G}(u,v) \leq d_{G-R_u}(u,V) \leq d_{G-R}(u,V)$$
	
	To evaluate the distance $d_{G-R}(u,v)$, note that it is equal to the distance between $u$ and $v$ in the network $G$ restricted on the edge set $B$. 
	Consider a shortest path $\pi_{G^*}(u,v)$ in the optimum network $G^*$. 
	By definition, for each edge $xy\in \pi_{G^*}(u,v)$, set $B$ contains a shortest path $\pi_G(x,y)$ of length $$d_{G-R}(x,y)\leq (\alpha+1)d_H(x,y)\leq (\alpha+1)d_{G^*}(x,y).$$
	Thus,  \begin{align*}
	d_{G-R}(u,v)&\leq \sum\limits_{xy\in \pi_{G^*}(u,v)}d_{G-R}(x,y)\\
	&\leq (\alpha+1)\cdot\sum\limits_{xy\in \pi_{G^*}(u,v)}d_{G^*}(x,y)=(\alpha+1)d_{G^*}(u,v).
	\end{align*}
	The total cost of set $R$ then is $$\alpha w(R)=\alpha\sum\limits_{u\in V}w(R_u)\leq \sum\limits_{u\in V}d_{G-R}(u,V)\leq(\alpha+1)\sum\limits_{u\in V}d_{G^*}(u,V).$$
	In combination with Inequality~(\ref{poa:gncg:firstsummand}) and the upper bound for the distance cost in Inequality~(\ref{poa:gncg:thirdsummand}), we get
	\begin{align*}
		\frac{\SC(G)}{\SC(G^*)}  &=\frac{\alpha w(B)+\alpha w(R) + \sum_{u\in V}d_G(u,V)}{\alpha w(E^*)+\sum_{u\in V}d_{G^*}(u,V)}                                        \\
		    & \leq\frac{\alpha(\alpha+1)w(E^*)+2(\alpha+1)\sum_{u\in V}d_{G^*}(u,V)}{\alpha w(E^*)+\sum_{u\in V}d_{G^*}(u,V)} \\
		    & \leq\frac{2(\alpha+1)(\alpha w(E^*)+\sum_{u\in V}d_{G^*}(u,V))}{\alpha w(E^*)+\sum_{u\in V}d_{G^*}(u,V)}                                 
		     =2(\alpha+1).\qedhere
	\end{align*} 
\end{proof}

\begin{corollary}\label{cor:PoA_GNCG}
$PoA\in\Theta(\alpha)$ in the GNCG.
\end{corollary}

\section{Conclusion}
We studied the problem of designing networks that are both efficient in terms of social cost and stable in terms of being close to a Nash equilibrium state. For this, we focus on studying $(\beta, \gamma)$-networks that are in $\beta$-approximate Nash equilibrium and have a social cost of at most $\gamma$ times the cost of the social optimum. In particular, we considered $(\beta,\gamma)$-networks in the Euclidean version of the Generalized Network Creation Game by Bilò et al.~\citep{bilo2019geometric}, where agents are points in $\mathbb{R}^d$, and each agent aims to maximize her centrality by creating costly edges. This version has the natural feature that the cost of each edge is proportional to the Euclidean distance between the endpoints. Hence, this model captures many real-world settings for the decentralized creation of communication networks.   

Our main contribution is a $\mathcal{O}(n^2)$ time algorithm for computing $(\beta,\beta)$-networks with low $\beta$. First of all, this result is interesting since it is one of the very few algorithmic results for constructing (approximate) Nash equilibria in the realm of network creation games. Such a centralized algorithm is valuable in a setting with strategic agents since a central designer could propose a network to the strategic agents which then may selfishly deviate from the proposed solution. If this proposed network is (almost) stable, then the agents have no (or only a very low) incentive for deviating. If additionally the proposed network has other beneficial properties like (almost) optimal social cost then this is another compelling reason for accepting the centrally designed proposal.

Second, our algorithm is simple but non-trivial and relies on techniques from the well-studied $t$-spanner problem. 
Moreover, our algorithm creates $(\mathcal{O}(1),\mathcal{O}(1))$-networks if $\alpha \leq \sqrt[3]{n}$, i.e., for the case where edges are comparably cheap or, even more realistic, where the number of nodes in the network is large. The same holds true for networks on random point sets or on grids, both of which seem to be natural topologies.

In contrast to these positive results, we observed that none of the extreme cases of $(\beta.\gamma)$-networks could guarantee a constant approximation. Namely, a social optimum network can be very unstable and may be NP-hard to compute, while a Nash equilibrium can have a much higher social cost than the optimal network. Moreover, we have shown that the finite improvement property does not hold for our model and it was shown by Bilò et al.~\cite{bilo2019geometric} that computing best response strategies is NP-hard. This indicates that there is no efficient way of finding an (almost) stable state in a decentralized way. Hence, computing such a state via a centralized algorithm and then proposing it to the agents could be a way to circumvent this hard problem.

As another important result of the paper, we have shown that the PoA depends only on the parameter $\alpha$ and not on the dimension of the Euclidean space or on metric edge weights. Although conjectured by Bilò et al.~\cite{bilo2019geometric}, this is surprising because it contradicts the intuition that the PoA should be lower for the metric case, especially for low dimensions.

We focused on three extreme cases of the bicriteria optimization, i.e., when one of the approximation factors is 1 or both approximation factors are equal. Of course, it would be interesting to map the whole Pareto frontier precisely. 
Another promising direction for future work is to extend our approximation results for the non-metric case. For example, any solution for a $k$-distributable $t$-spanner for weighted networks would make our approximation algorithm work even for the Generalized Network Creation Game.

\bibliographystyle{abbrv}
\bibliography{paper}

\newcommand{\SortNoop}[1]{}
\begin{thebibliography}{10}

\bibitem{Al06}
S.~Albers, S.~Eilts, E.~Even-Dar, Y.~Mansour, and L.~Roditty.
\newblock On {Nash} equilibria for a network creation game.
\newblock In {\em {SODA '06}}, pages 89--98, 2006.

\bibitem{AL10}
S.~Albers and P.~Lenzner.
\newblock On approximate {Nash} equilibria in network design.
\newblock {\em Internet Mathematics}, 9(4):384--405, 2013.

\bibitem{ADHL13}
N.~Alon, E.~D. Demaine, M.~T. Hajiaghayi, and T.~Leighton.
\newblock Basic network creation games.
\newblock {\em SIAM Journal on Discrete Mathematics}, 27(2):656--668, 2013.

\bibitem{AM17}
C.~{\`A}lvarez and A.~Messegu{\'e}.
\newblock Network creation games: Structure vs anarchy.
\newblock {\em arXiv:1706.09132}, 2017.

\bibitem{AM18}
C.~{\`A}lvarez and A.~Messegu{\'e}.
\newblock On the constant price of anarchy conjecture.
\newblock {\em arXiv:1809.08027}, 2018.

\bibitem{AM19}
C.~{\`A}lvarez and A.~Messegu{\'e}.
\newblock On the price of anarchy for high-price links.
\newblock In {\em {WINE'19}}, pages 316--329, 2019.

\bibitem{ADKTWR}
E.~Anshelevich, A.~Dasgupta, J.~Kleinberg, E.~Tardos, T.~Wexler, and
  T.~Roughgarden.
\newblock The price of stability for network design with fair cost allocation.
\newblock {\em SIAM Journal on Computing}, 38(4):1602--1623, 2008.

\bibitem{ADTW08}
E.~Anshelevich, A.~Dasgupta, {\'{E}}.~Tardos, and T.~Wexler.
\newblock Near-optimal network design with selfish agents.
\newblock {\em Theory of Computing}, 4(1):77--109, 2008.

\bibitem{BG00}
V.~Bala and S.~Goyal.
\newblock A noncooperative model of network formation.
\newblock {\em Econometrica}, 68(5):1181--1229, 2000.

\bibitem{BFLLM20}
D.~Bil{\`{o}}, T.~Friedrich, P.~Lenzner, S.~Lowski, and A.~Melnichenko.
\newblock Selfish creation of social networks.
\newblock {\em CoRR}, abs/2012.06203, 2020.

\bibitem{bilo2019geometric}
D.~Bil{\`o}, T.~Friedrich, P.~Lenzner, and A.~Melnichenko.
\newblock Geometric network creation games.
\newblock In {\em {SPAA'19}}, pages 323--332, 2019.

\bibitem{BFLMM20}
D.~Bil{\`{o}}, T.~Friedrich, P.~Lenzner, A.~Melnichenko, and L.~Molitor.
\newblock Fair tree connection games with topology-dependent edge cost.
\newblock In {\em {FSTTCS'20}}, pages 15:1--15:15, 2020.

\bibitem{BL18}
D.~Bilò and P.~Lenzner.
\newblock On the tree conjecture for the network creation game.
\newblock In {\em {STACS'18}}, pages 14:1--14:15, 2018.

\bibitem{BK11}
M.~Brautbar and M.~J. Kearns.
\newblock A clustering coefficient network formation game.
\newblock In {\em {SAGT'11}}, pages 224--235, 2011.

\bibitem{carmi2013minimum}
P.~Carmi and L.~Chaitman-Yerushalmi.
\newblock Minimum weight euclidean t-spanner is np-hard.
\newblock {\em Journal of Discrete Algorithms}, 22:30--42, 2013.

\bibitem{CLMM17}
A.~Chauhan, P.~Lenzner, A.~Melnichenko, and L.~Molitor.
\newblock Selfish network creation with non-uniform edge cost.
\newblock In {\em {SAGT'17}}, pages 160--172, 2017.

\bibitem{CLMM16}
A.~Chauhan, P.~Lenzner, A.~Melnichenko, and M.~M{\"{u}}nn.
\newblock On selfish creation of robust networks.
\newblock In {\em {SAGT'16}}, pages 141--152, 2016.

\bibitem{CW18}
S.~Chechik and C.~Wulff{-}Nilsen.
\newblock Near-optimal light spanners.
\newblock {\em {ACM} Trans. Algorithms}, 14(3):33:1--33:15, 2018.

\bibitem{CR09}
H.~Chen and T.~Roughgarden.
\newblock Network design with weighted players.
\newblock {\em Theory Comput. Syst.}, 45(2):302--324, 2009.

\bibitem{CRV08}
H.~Chen, T.~Roughgarden, and G.~Valiant.
\newblock Designing networks with good equilibria.
\newblock In {\em {SODA'08}}, pages 854--863, 2008.

\bibitem{CL15}
A.~Cord{-}Landwehr and P.~Lenzner.
\newblock Network creation games: Think global - act local.
\newblock In {\em {MFCS'15}}, pages 248--260, 2015.

\bibitem{CSSM04}
J.~R. Correa, A.~S. Schulz, and N.~E. Stier-Moses.
\newblock Selfish routing in capacitated networks.
\newblock {\em Mathematics of Operations Research}, 29(4):961--976, 2004.

\bibitem{De07}
E.~D. Demaine, M.~T. Hajiaghayi, H.~Mahini, and M.~Zadimoghaddam.
\newblock The price of anarchy in network creation games.
\newblock {\em ACM Trans. on Algorithms}, 8(2):13, 2012.

\bibitem{EFLM20}
H.~Echzell, T.~Friedrich, P.~Lenzner, and A.~Melnichenko.
\newblock Flow-based network creation games.
\newblock In {\em {IJCAI'20}}, pages 139--145, 2020.

\bibitem{EKZ06}
S.~Eidenbenz, V.~A. Kumar, and S.~Zust.
\newblock Equilibria in topology control games for ad hoc networks.
\newblock {\em Mobile Networks and Applications}, 11(2):143--159, 2006.

\bibitem{ENS15}
M.~Elkin, O.~Neiman, and S.~Solomon.
\newblock Light spanners.
\newblock {\em {SIAM} J. Discret. Math.}, 29(3):1312--1321, 2015.

\bibitem{Fab03}
A.~Fabrikant, A.~Luthra, E.~Maneva, C.~H. Papadimitriou, and S.~Shenker.
\newblock On a network creation game.
\newblock In {\em {PODC'03}}, pages 347--351, 2003.

\bibitem{FS20}
A.~Filtser and S.~Solomon.
\newblock The greedy spanner is existentially optimal.
\newblock {\em {SIAM} J. Comput.}, 49(2):429--447, 2020.

\bibitem{Fried17}
T.~Friedrich, S.~Ihde, C.~Ke{\ss}ler, P.~Lenzner, S.~Neubert, and D.~Schumann.
\newblock Efficient best response computation for strategic network formation
  under attack.
\newblock In {\em {SAGT'17}}, pages 199--211, 2017.

\bibitem{GairingHK14}
M.~Gairing, T.~Harks, and M.~Klimm.
\newblock Complexity and approximation of the continuous network design
  problem.
\newblock In {\em {APPROX/RANDOM'14}}, pages 226--241, 2014.

\bibitem{GJ02}
M.~R. Garey and D.~S. Johnson.
\newblock {\em Computers and intractability}, volume~29.
\newblock wh freeman New York, 2002.

\bibitem{Gul15}
A.~Guly{\'a}s, J.~J. Bir{\'o}, A.~K{\H{o}}r{\"o}si, G.~R{\'e}tv{\'a}ri, and
  D.~Krioukov.
\newblock Navigable networks as {Nash} equilibria of navigation games.
\newblock {\em Nature communications}, 6:7651, 2015.

\bibitem{gupta2011approximation}
A.~Gupta and J.~K{\"o}nemann.
\newblock Approximation algorithms for network design: A survey.
\newblock {\em Surveys in Operations Research and Management Science},
  16(1):3--20, 2011.

\bibitem{H09}
M.~Hoefer.
\newblock Non-cooperative tree creation.
\newblock {\em Algorithmica}, 53(1):104--131, 2009.

\bibitem{HK05}
M.~Hoefer and P.~Krysta.
\newblock Geometric network design with selfish agents.
\newblock In {\em {COCOON'05}}, pages 167--178, 2005.

\bibitem{JLK78}
D.~S. Johnson, J.~K. Lenstra, and A.~R. Kan.
\newblock The complexity of the network design problem.
\newblock {\em Networks}, 8(4):279--285, 1978.

\bibitem{karp1972reducibility}
R.~M. Karp.
\newblock Reducibility among combinatorial problems.
\newblock In {\em Complexity of computer computations}, pages 85--103.
  Springer, 1972.

\bibitem{KL13}
B.~Kawald and P.~Lenzner.
\newblock On dynamics in selfish network creation.
\newblock In {\em {SPAA'13}}, pages 83--92, 2013.

\bibitem{KP99}
E.~Koutsoupias and C.~Papadimitriou.
\newblock Worst-case equilibria.
\newblock In {\em {STACS'99}}, pages 404--413, 1999.

\bibitem{LS19}
H.~Le and S.~Solomon.
\newblock Truly optimal euclidean spanners.
\newblock In {\em {FOCS'19}}, pages 1078--1100, 2019.

\bibitem{L11}
P.~Lenzner.
\newblock On dynamics in basic network creation games.
\newblock In {\em SAGT'11}, pages 254--265. 2011.

\bibitem{Len12}
P.~Lenzner.
\newblock Greedy selfish network creation.
\newblock In {\em {WINE'12}}, pages 142--155, 2012.

\bibitem{MW84}
T.~L. Magnanti and R.~T. Wong.
\newblock Network design and transportation planning: Models and algorithms.
\newblock {\em Transportation science}, 18(1):1--55, 1984.

\bibitem{MMM15}
A.~Mamageishvili, M.~Mihal{\'{a}}k, and D.~M{\"{u}}ller.
\newblock Tree {Nash} equilibria in the network creation game.
\newblock {\em Internet Mathematics}, 11(4-5):472--486, 2015.

\bibitem{MS12}
M.~Mihal\'{a}k and J.~C. Schlegel.
\newblock Asymmetric swap-equilibrium: A unifying equilibrium concept for
  network creation games.
\newblock In {\em {MFCS'12}}, pages 693--704, 2012.

\bibitem{MS13}
M.~Mihal{\'{a}}k and J.~C. Schlegel.
\newblock The price of anarchy in network creation games is (mostly) constant.
\newblock {\em Theory Comput. Syst.}, 53(1):53--72, 2013.

\bibitem{MS96}
D.~Monderer and L.~S. Shapley.
\newblock Potential games.
\newblock {\em Games and Economic Behavior}, 14(1):124 -- 143, 1996.

\bibitem{MSW11}
T.~Moscibroda, S.~Schmid, and R.~Wattenhofer.
\newblock Topological implications of selfish neighbor selection in
  unstructured peer-to-peer networks.
\newblock {\em Algorithmica}, 61(2):419--446, 2011.

\bibitem{narasimhan2007geometric}
G.~Narasimhan and M.~Smid.
\newblock {\em Geometric spanner networks}.
\newblock Cambridge University Press, 2007.

\bibitem{newman10}
M.~Newman.
\newblock {\em Networks: an introduction}.
\newblock Oxford university press, 2010.

\end{thebibliography}

\end{document}